\documentclass[onecolumn]{IEEEtran}
\usepackage{amsthm}
\usepackage{times,amssymb,amsmath,amsfonts,float,nicefrac,color,bbm,mathrsfs,float}
\usepackage{enumerate,multirow,tikz,graphicx,enumitem,soul}
\usepackage[mathscr]{eucal}
\usepackage{sidecap}
\usepackage{algpseudocode}
\usepackage{verbatim}
\usepackage{epstopdf}
\usepackage{textcomp}
\usetikzlibrary{shapes,arrows}
\usepackage{mathabx}
\usepackage[noadjust]{cite}
\usepackage{booktabs}
\usepackage[top=1in,bottom=1in,left=1.2in,right=1.2in]{geometry}
\allowdisplaybreaks

\usepackage[bottom]{footmisc}
\usepackage[normalem]{ulem}

\usepackage[ruled,vlined,linesnumbered]{algorithm2e}

\usepackage{hyperref}

\newcommand\nc\newcommand
\nc\bfa{{\boldsymbol a}}\nc\bfA{{\boldsymbol A}}\nc\cA{{\mathscr A}}
\nc\bfb{{\boldsymbol b}}\nc\bfB{{\boldsymbol B}}\nc\cB{{\mathscr B}}
\nc\bfc{{\boldsymbol c}}\nc\bfC{{\boldsymbol C}}\nc\cC{{\mathscr C}}
\nc\bfd{{\boldsymbol d}}\nc\bfD{{\boldsymbol D}}\nc\cD{{\mathscr D}}
\nc\bfe{{\boldsymbol e}}\nc\bfE{{\boldsymbol E}}\nc\cE{{\mathscr E}}
\nc\bff{{\boldsymbol f}}\nc\bfF{{\boldsymbol F}}\nc\cF{{\mathscr F}}
\nc\bfg{{\boldsymbol g}}\nc\bfG{{\boldsymbol G}}\nc\cG{{\mathscr G}}
\nc\bfh{{\boldsymbol h}}\nc\bfH{{\boldsymbol H}}\nc\cH{{\mathscr H}}
\nc\bfi{{\boldsymbol i}}\nc\bfI{{\boldsymbol I}}\nc\cI{{\mathcal I}}
\nc\bfj{{\boldsymbol j}}\nc\bfJ{{\boldsymbol J}}\nc\cJ{{\mathscr J}}
\nc\bfk{{\boldsymbol k}}\nc\bfK{{\boldsymbol K}}\nc\cK{{\mathscr K}}
\nc\bfl{{\boldsymbol l}}\nc\bfL{{\boldsymbol L}}\nc\cL{{\mathscr L}}
\nc\bfm{{\boldsymbol m}}\nc\bfM{{\boldsymbol M}}\nc{\cM}{{\mathscr M}}
\nc\bfn{{\boldsymbol n}}\nc\bfN{{\boldsymbol N}}\nc\cN{{\mathscr N}}
\nc\bfo{{\boldsymbol o}}\nc\bfO{{\boldsymbol O}}\nc\cO{{\mathscr O}}
\nc\bfp{{\boldsymbol p}}\nc\bfP{{\boldsymbol P}}\nc\cP{{\mathscr P}}\nc\fP{{\mathfrak P}}
\nc\bfq{{\boldsymbol q}}\nc\bfQ{{\boldsymbol Q}}\nc\cQ{{\mathscr Q}}
\nc\bfr{{\boldsymbol r}}\nc\bfR{{\boldsymbol R}}\nc\cR{{\mathscr R}}
\nc\bfs{{\boldsymbol s}}\nc\bfS{{\boldsymbol S}}\nc\cS{{\mathscr S}}
\nc\bft{{\boldsymbol t}}\nc\bfT{{\boldsymbol T}}\nc\cT{{\mathscr T}}
\nc\bfu{{\boldsymbol u}}\nc\bfU{{\boldsymbol U}}\nc\cU{{\mathscr U}}
\nc\bfv{{\boldsymbol v}}\nc\bfV{{\boldsymbol V}}\nc\cV{{\mathscr V}}
\nc\bfw{{\boldsymbol w}}\nc\bfW{{\boldsymbol W}}\nc\cW{{\mathscr W}}
\nc\bfx{{\boldsymbol x}}\nc\bfX{{\boldsymbol X}}\nc\cX{{\mathscr X}}
\nc\bfy{{\boldsymbol y}}\nc\bfY{{\boldsymbol Y}}\nc\cY{{\mathscr Y}}
\nc\bfz{{\boldsymbol z}}\nc\bfZ{{\boldsymbol Z}}\nc\cZ{{\mathscr Z}}

\newtheorem{theorem}{Theorem}
\newtheorem{lemma}[theorem]{Lemma}

\newtheorem{proposition}[theorem]{Proposition}

\newtheorem{corollary}[theorem]{Corollary}
\newtheorem{definition}{Definition}

\newtheorem{construction}{Construction}[section]

\theoremstyle{remark}

\DeclareMathOperator{\rank}{rank}

\DeclareMathOperator{\Span}{Span}
\newcommand{\ff}{{\mathbb F}}

\DeclareMathOperator{\acs}{Acs}
\newcommand{\tr}{\mathrm{tr}}
\newcommand{\zitan}[1]{#1}

\begin{document}
	
		\title{Rack-aware minimum-storage regenerating codes with optimal access}
	
	\author{
		\IEEEauthorblockN{Jiaojiao Wang} \hspace*{1in}
		\and
		\IEEEauthorblockN{Zitan Chen}}
	\maketitle	
	
	{\renewcommand{\thefootnote}{}\footnotetext{
			
			\vspace{-.2in}
			
			\noindent\rule{1.5in}{.4pt}

			{	
				This paper was presented in part at the 2022 IEEE Information Theory Workshop (ITW 2022) \cite{chen2022rack} and will be presented in part at the 2023 IEEE International Symposium on Information Theory (ISIT 2023). 
				
				J. Wang is with the Data Science and Information Technology Research Center, Tsinghua-Berkeley Shenzhen Institute, Tsinghua Shenzhen International Graduate School, Shenzhen, China. (Email: wjj22@mails.tsinghua.edu.cn)			
				
				Z. Chen is with the School of Science and Engineering, the Guangdong Provincial Key Laboratory of Future Networks of Intelligence, The Chinese University of Hong Kong, Shenzhen, China. (Email: chenztan@cuhk.edu.cn)
				His research was supported in part by the Basic Research Project of Hetao
				Shenzhen-Hong Kong Science and Technology Cooperation Zone under
				Project HZQB-KCZYZ-2021067, the Guangdong Provincial Key Laboratory of Future Network of Intelligence under 
				Project 2022B1212010001, and the National Natural Science Foundation of China under grants 62201487 and 12141108.
			}
		}
	}
	\renewcommand{\thefootnote}{\arabic{footnote}}
	\setcounter{footnote}{0}

	\begin{abstract} We derive a lower bound on the amount of information accessed to repair \zitan{failed nodes within a single rack} from any number of helper racks in the rack-aware storage model that allows collective information processing in the nodes that share the same rack. \zitan{Furthermore, we} construct a family of rack-aware minimum-storage regenerating (MSR) codes with \zitan{the property that} the number of symbols accessed for \zitan{repairing a single failed node} attains the bound with equality for all admissible parameters. Constructions of rack-aware optimal-access MSR codes were only known for limited parameters. \zitan{We also present a family of Reed-Solomon (RS) codes that only require accessing a relatively small number of symbols to repair multiple failed nodes in a single rack. In particular, for certain code parameters, the RS construction attains the bound on the access complexity with equality and thus has optimal access.}
	\end{abstract}

\section{Introduction}
The rapid development of distributed storage systems raised the question of how failed nodes in these systems can be efficiently repaired. A large body of literature on erasure codes for distributed storage addressed this question over the past decade. One approach to assess repair efficiency is to measure the so-called repair bandwidth, which is the amount of information downloaded from other nodes for the repair. This approach, first introduced in \cite{dimakis2010network}, assumes the nodes form a homogeneous network and determines the repair bandwidth by considering the information flow in course of the repair. In particular, \cite{dimakis2010network} derived a bound on the smallest number of symbols required for the repair of a single failed node, known as the cut-set bound on the repair bandwidth. Furthermore, the results of \cite{dimakis2010network} showed that there is a trade-off between the minimum repair bandwidth and the storage capacity of the network. In this paper we shall focus on codes with minimum storage overhead and optimal repair bandwidth, namely, maximum distance separable (MDS) codes with optimal repair bandwidth. Such codes are termed minimum-storage regenerating (MSR) codes in the literature.

The problem of MSR codes has received much attention and affords several variants. In its basic form, the problem concerns repair of a single failed node, but it has been generalized to repair of multiple failed nodes. There are generally two different models for repairing multiple failed nodes: the centralized repair model \cite{cadambe2013asymptotic}, \cite{hu2015broadcast}, \cite{ye2017explicit}, \cite{rawat2018centralized} and the cooperative repair model \cite{shum2013cooperative}, \cite{kermarrec2011repairing}, \cite{li2014cooperative}, \cite{ye2018cooperative}. The centralized model assumes that the failed nodes are repaired by a single data collector that received information from the helper nodes and performs the recovery within a single location. In contrast, the cooperative model assumes that all the failed nodes are restored at distinct physical locations and the exchange of information between these locations, in addition to the information downloaded to each of them, is counted toward the overall repair bandwidth. A large number of papers have been devoted to the study of explicit constructions of MSR codes for repairing a single failed node as well as centralized and cooperative repair for multiple failed nodes, including \cite{rashmi2011optimal}, \cite{tamo2012zigzag}, \cite{ye2017explicit}, \cite{sasidharan2016explicit}, \cite{li2018generic}, \cite{elyasi2020cascade}, \cite{ye2018cooperative}, \cite{duursma2021multi}, among others.

Although the problem of MSR codes was initially formulated for homogeneous storage models, various extensions to models of clustered architectures have been proposed in an effort to better capture characteristics of storage systems in the real world. These extensions typically assume the nodes are formed into groups and distinguish the bandwidth cost for communication within a group from the cost incurred when communicating across different groups \cite{akhlaghi2010cost}, \cite{gaston2013realistic}, \cite{pernas2013non}, \cite{sohn2018capacity}. While the aforementioned papers assume the communication cost depends on the memberships of the communicating nodes, they do not allow nodes within the same group to pool their data and collectively process the information before transmitting it to the failed node. A clustered storage model that takes into account the possibility of collective processing of group data was proposed in \cite{hu2016double}, also known as the \emph{rack-aware storage} model. More precisely, the rack model assumes that the nodes are organized into groups of equal size, also called racks. Suppose a node has failed and call the rack that contains the failed node the \emph{host rack}. The repair of the failed node is accomplished by downloading information from nodes in the host rack, called \emph{local nodes}, and information from \zitan{some} other racks. Besides allowing the nodes that share the same helper rack to jointly process the data, the model further assumes that communication within each rack, including the host rack, does not incur any cost toward the repair bandwidth. In this paper we limit ourselves to MSR codes for rack-aware storage, namely, rack-aware MSR codes. Before proceeding to discuss prior work on the rack model in more detail, we mention that there are other variations of clustered architectures such as \cite{tebbi2014code}, \cite{ye2017explicit}, \cite{sahraei2018increasing}, \cite{prakash2018storage} and extensions to general connectivity constraints represented by a graph including \cite{gerami2011optimal}, \cite{lu2014distributed}, \cite{patra2021node}.

A version of the cut-set bound for a single failed node in the rack model was derived in \cite{hu2016double}, \cite{hou2019rack}. Moreover, from the cut-set \zitan{bounds for the rack model and the homogeneous model}, it was observed that rack-aware coding is \zitan{at least as good as homogeneous coding in terms of repair bandwidth}. More precisely, suppose that the storage system consists of $\bar{n}$ racks of the same size $u$ and that $k$ data blocks are encoded into a codeword of length $n=\bar{n}u$ by rack-aware coding, stored across $n$ nodes. If the remainder $v:=k\bmod u$ is nonzero then rack-aware coding \zitan{can have} strictly smaller repair bandwidth than \zitan{rack-oblivious homogeneous coding}. Existence of rack-aware MSR codes were shown \cite{hu2016double}, \cite{hou2019rack}. Constructions of such codes is presented in \cite{hu2017optimal} for $3$ racks and for the case when the number of parity symbols $r:=n-k$ is equal to the rack size $u$. Explicit constructions of rack-aware MSR codes were first provided in \cite{chen2019explicit}, covering all admissible parameters such as the code rate $k/n$, the rack size $u$, and the number of racks $\bar{n}$. The main idea in \cite{chen2019explicit} that enables explicit constructions of rack-aware MSR codes is to utilize the multiplicative structure of the underlying finite field and align it with the group structure of the rack model. Later, \cite{hou2020minimum} \zitan{and} \cite{zhou2021explicit} gave constructions of rack-aware MSR codes that have smaller node size compared \zitan{with} the constructions in \cite{chen2019explicit}. Specifically, \cite{hou2020minimum} proposed a coding framework that converts MSR codes to rack-aware MSR codes based on \zitan{the Schwartz-Zippel lemma, and} the construction with reduced node size in \cite{zhou2021explicit} is found by an algorithmic approach that \zitan{explicitly} determines the parity-check equations for the code.

Apart from the bandwidth, the repair efficiency is impacted by the number of symbols accessed on the helper racks to generate the information to be download by the failed node. For the homogeneous model, some constructions such as \cite{tamo2014access}, \cite{ye2017access}, \cite{zhang2020explicit}, \cite{vajha2021small} have the property that the amount of information accessed on the helper nodes is the smallest possible, thus called \emph{optimal-access} MSR codes. In \cite{chen2019explicit}, the authors derived a lower bound on the number of symbols required to be accessed on the helper racks to repair the failed node. Moreover, \cite{chen2019explicit} presented a construction of rack-aware MSR codes with low access. \zitan{While this low-access construction had the smallest access among the known constructions, it} falls short of attaining the bound \zitan{in \cite{chen2019explicit}}. A new bound on the number of accessed symbols for the case when the number of helper racks is $\bar{n}-1$ was shown in a recent paper \cite{li2021repair}. Interestingly, this new bound implies that the low-access construction in \cite{chen2019explicit} is actually an optimal-access construction when the remainder $v$ of the code dimension $k$ divided by the rack size $u$ is equal to $u-1$. \zitan{However, the problem of constructing rack-aware MSR codes with the smallest access complexity for other code parameters remained open.}

Initially, known constructions of MSR codes such as \cite{rashmi2011optimal}, \cite{ye2017explicit}, \cite{rawat2018centralized} are array codes (or vector codes). The question of constructing \emph{scalar} MSR codes, especially RS codes, with optimal repair bandwidth attracted much interest due to their applications in classical and modern storage systems. Constructions of optimal-repair RS codes for homogeneous storage were studied in \cite{guruswami2017repairing}, \cite{dau2017optimal}, \cite{tamo2018repair}. Moreover, optimal-access RS codes were constructed in \cite{chen2020enabling}. For the rack model, a family of RS codes with optimal bandwidth for repairing a single failed node was presented in \cite{chen2019explicit}. Studies of RS codes with optimal repair bandwidth for the rack model were also presented in \cite{jin2019optimal}, \cite{wang2021rack}. For the access complexity of repairing RS codes in the rack model, it was observed in \cite[Section~3.5.2]{chen2020codes} that the RS codes constructed in \cite{chen2019explicit} can be modified to obtain rack-aware RS codes with low access by incorporating ideas developed in \cite{chen2020enabling} for constructing optimal-access RS codes in the homogeneous model. According to the recent bound in \cite{li2021repair}, this observation in fact leads to a construction of rack-aware RS code with optimal access for the case when the number of helper racks is $\bar{n}-1$ and $v=u-1$. However, the details of this result have not been formally documented. At the same time, there are no general results of rack-aware RS codes with optimal access. 

We note that the problem of repairing \emph{multiple} failed nodes in rack-aware storage was explored in several recent papers \cite{gupta2020rack}, \cite{zhou2021rack}, \cite{wang2021rack}. However, the known results are far from forming a conclusive picture for the problem. 
In particular, to the best of our knowledge, no bounds are known for the access complexity of repairing multiple nodes in rack-aware storage.

\subsection{Main results}
In this paper we present a lower bound on the number of symbols accessed to repair \zitan{multiple failed nodes within a single rack} from any number of helper racks, thereby showing that the low-access construction in \cite{chen2019explicit} is, in fact, an optimal-access construction for repairing a single failed node from any number of helper racks when $v=u-1$. This is the first lower bound for the access complexity of repairing multiple nodes from any number of helper racks. 

Our second result is a family of rack-aware MSR codes that attains the bound we derive \zitan{for repairing a single node} with equality for any $0\leq v<u$. The field size required for the construction is linear in $n$ and the node size is exponential in $\bar{n}$. As a matter of fact, this new family of codes coincides with the low-access construction in \cite{chen2019explicit} if $v=u-1$. 

Lastly, we also present a family of RS codes that supports optimal repair of $h\leq u-v$ failed nodes within a single rack from any number of helper racks. The amount information accessed for repair of these codes is within a factor of $(u-v)/h$ of the lower bound we derive. In particular, when the number of failed nodes in the host rack is exactly $u-v$, this construction gives a family of rack-aware scalar MSR codes with optimal access.

\subsection{Organization}
In section~\ref{sec:ps} we set up notation and formulate the rack-aware storage model for the case of multiple failed nodes in a single host rack. Section~\ref{sec:bound} is devoted to the lower bound on the access complexity of optimal repair schemes. The family of rack-aware MSR array codes with optimal access for a single node failure is presented in Section~\ref{sec:construction} and the family of RS codes with low access for multiple node failures in a single rack is given Section~\ref{sec:rs}. We conclude in Section~\ref{sec:con} with open questions. 

\section{Problem statement}
\label{sec:ps}

Consider an $(n,k,l)$ array code $\cC$ over a finite field $F$. Specifically, $\cC$ is \zitan{a collection} of codewords $c=(c_0,\ldots,c_{n-1})$ with $c_j=(c_{j,0},\ldots,c_{j,l-1})^T\in F^l, j=0,\ldots,n-1$. We assume that the code $\cC$ forms a linear subspace over $F$. \zitan{If $\cC$ is linear over $F^l$, then $\cC$ is called a \emph{scalar} code to stress the linearity property.} Moreover, we assume $\cC$ is MDS, namely, each codeword in $\cC$ can be recovered from any $k$ of its coordinates. Suppose that $n=\bar{n}u$ where $u>1$ is an integer and $k=\bar{k}u+v$ where $v = k \bmod u$. To rule out the trivial case, we assume throughout that $k\geq u$, i.e., $\bar{k}\geq 1$. Denote by $r:=n-k$ the number of parity symbols.
		
Assume that the data file of size $kl$ is divided into $k$ blocks, encoded into a codeword $c\in\cC$, \zitan{and stored in $n$ nodes}. The set of nodes $\{0,\ldots,n-1\}$ is partitioned into $\bar{n}$ subsets of size $u$, i.e., racks of size $u$. Accordingly, the coordinates of the codeword $c$ are partitioned into segments of length $u$, and we label them as $c_j,j=0,\ldots,n-1$ where $j=eu+g$, $e=0,\ldots,\bar{n}-1$ and $g=0,\ldots,u-1$. We do not distinguish between nodes and the coordinates of the codeword, and refer to both of them as nodes. For notational convenience, we write $\bar{c}_e:=(c_{eu},\ldots,c_{eu+u-1})$ to denote the content of the $e$th rack.

\subsection{Optimal repair}

Denote by $\cR\subset \{0,\ldots,\bar{n}-1\}$ the set of helper racks and let $\bar{d}=|\cR|,\bar{k}\leq \bar{d}\leq\bar{n}-1$. Further, \zitan{let $e'$ be the index of the host rack and let $\cF=\{j'_1=e'u+g'_1,\ldots,j'_h=e'u+g'_h\}$ be the set of the failed nodes where $h\geq 1$}. To recover the failed nodes, information is generated from the symbols in the helper racks, i.e., $\bar{c}_{e},e\in\cR$, as well as the the contents of the local nodes $c_{e'u+g},g\in\{0,\ldots,u-1\}\setminus\{\zitan{g'_1,\ldots,g'_h}\}$. A repair scheme with degree $\bar{d}$ is formed of $\bar{d}$ functions $f_t\colon F^{ul}\to F^{\beta_t},t=1,\ldots,\bar{d}$ and a function $g\colon F^{\sum_{t=1}^{\bar{d}}\beta_t}\times F^{(u-\zitan{h})l} \to F^{\zitan{hl}}$. For each $e_t\in\cR$ the function $f_t$ maps $\bar{c}_{e_t}$ to some $\beta_t$ symbols of $F$. The function $g$ accepts these symbols \zitan{from all the helper racks} and the contents of the local nodes in the host rack as arguments, and return the \zitan{values of the failed nodes}:
\begin{align*}
	g(\{f_t(\bar{c}_{e_t}),e_t\in\cR\},\{c_{e'u+g},g\in\{0,\ldots,u-1\}\setminus\{\zitan{g'_1,\ldots,g'_h}\}\}) = \zitan{\{c_{j'_1},\ldots,c_{j'_h}\}} \text{ for all } c\in\cC.
\end{align*} If the functions $f_t,g$ are $F$-linear, the repair scheme is said to be linear. The quantity $\beta(\cR,\zitan{\cF})=\sum_{e\in\cR}\beta_e$ is called the repair bandwidth of recovering the failed nodes \zitan{in $\cF$} from the helper racks in $\cR$.

Let 
\begin{align*}
	\beta(\bar{d},\zitan{h}) := \min_{\cC\subset F^{nl}} \max_{\cR,\zitan{\cF}} \beta(\cR,\zitan{\cF})
\end{align*} where the minimum is taken over all $(n,M=|F|^{kl})$ MDS codes over $F$ (including codes that are not $F$-linear) and the maximum over \zitan{all possible sets of helper racks $\cR$ of size $\bar{d}$ and sets of failed nodes $\cF$ of size $h$}.\footnote{Note that we restrict the failed nodes to be in a single rack.}

A necessary condition for successful repair of a single node is given by \cite{hou2019rack}, \cite{hu2016double}, which states that for any $(n,M=|F|^{kl})$ MDS code, the repair bandwidth satisfies
\begin{align}
	\beta(\bar{d},\zitan{1})\geq \frac{\bar{d}l}{\bar{d}-\bar{k}+1}.\label{eq:rack-bound}
\end{align}
\zitan{For recovery of multiple nodes to be possible, it is clear that the number of failed nodes should satisfy $h\leq \min\{u,r\}$. A more careful inspection of the parameters reveals that \[h\leq \min\{u,(\bar{d}-\bar{k}+1)u-v\}.\] Indeed, by the MDS property, the number of \emph{effective} helper nodes should be at least $k$ and thus $\bar{d}u+u-h\geq k$, implying $h\leq (\bar{d}-\bar{k}+1)u-v\leq r$. A version of the cut-set bound of the repair bandwidth for $h$ failed nodes was mentioned in \cite{chen2019explicit} without proof, which we restated below. 
\begin{proposition}
	\label{prop:rack-bound-h}
	For any $(n,k,l)$ MDS array code over $F$, the repair bandwidth of recovering $h\leq \min\{u,(\bar{d}-\bar{k}+1)u-v\}$ failed nodes from $\bar{d}\geq\bar{k}$ helper racks satisfies 
	\begin{align}
		\beta(\bar{d},h)\geq \frac{h\bar{d}l}{\bar{d}-\bar{k}+1}.\label{eq:rack-bound-h}
	\end{align}
	Furthermore, if $\bar{k} > 1$ then \eqref{eq:rack-bound-h} holds with equality if and only if each helper rack contributes $hl/(\bar{d}-\bar{k}+1)$ symbols of $F$ for the repair of the failed nodes.
\end{proposition}
A proof for Proposition~\ref{prop:rack-bound-h} is provided in the Appendix~\ref{app:rack-bound-h} for completeness.}
The MDS codes that attain the bound \eqref{eq:rack-bound-h} with equality are said to support optimal repair and called rack-aware MSR codes.

\subsection{Optimal access}
In general, the function $f_t$ may have to read all the symbols in its input $\bar{c}_t$ to produce the information required for the repair of the failed nodes, which impacts upon the communication complexity and repair efficiency. Thus, it is desirable to construct codes that support optimal repair and low access. A lower bound on the number of symbols accessed for repair was given in \cite{chen2019explicit}.

\begin{proposition}[\cite{chen2019explicit}]\label{prop:node-access}
	Let $\cC$ be an $(n,k,l)$ MDS array code that supports optimal repair for a single node from $\bar{d}\geq \bar{k}+1$ helper racks where $\bar{k}\geq 1$. The number of symbols accessed on the helper racks satisfies 
	\begin{align}
		\alpha \geq \frac{\bar{d}ul}{d-k+1},\label{eq:rack-access-smallest}
	\end{align} where $d=\bar{d}u+u-1$. If $k>1$ then equality holds if and only if the number of symbols accessed on node $j=eu+g$ satisfies $\alpha_j=l/(d-k+1)$ for all $e\in\cR$ and $g=0,\ldots,u-1$.
\end{proposition}

This bound \eqref{eq:rack-access-smallest} was later improved in \cite{li2021repair} for the case when $\bar{d}=\bar{n}-1$.

\begin{proposition}[\cite{li2021repair}]\label{prop:access-ruv}
	Let $2\leq \bar{k}\leq \bar{n}-2$ and let $\cC$ be an $(n,k,l)$ MDS array code that supports optimal repair for a single node from $\bar{n}-1$ helper racks. The number of symbols accessed on the helper racks satisfies 
	\begin{align}
		\alpha \geq \frac{\bar{d}ul}{(\bar{n}-\bar{k})(u-v)}.\label{eq:rack-access-all}
	\end{align} Equality holds if and only if the number of symbols accessed on node $j=eu+g$ satisfies $\alpha_j=l/((\bar{n}-\bar{k})(u-v))$ for all $e\in\cR$ and $g=0,\ldots,u-1$.
\end{proposition}

We note that the number of symbols accessed on the helper racks for the rack-aware MSR codes in \cite[Section~IV]{chen2019explicit} is equal to $\bar{d}ul/(\bar{d}-\bar{k}+1)$. Therefore, according to \eqref{eq:rack-access-all}, they form a class of rack-aware MSR codes with the least access for repairing a single node from $\bar{d}=\bar{n}-1$ helper racks if $v=u-1$.

\zitan{In Section~\ref{sec:bound}, we derive a new lower bound on the number of accessed symbols for any optimal-bandwidth linear repair scheme with degree $\bar{d}$ and formally define the notion of optimal access for rack-aware MSR codes according to the new bound. In Section~\ref{sec:construction}, we construct a family of rack-aware MSR codes with optimal access for all admissible parameters, where the underlying finite field is of size linear in $n$ and the node size $l$ is exponential in $\bar{n}$.}
In the sequel, we denote $\bar{r}:=\bar{n}-\bar{k}$ and $\bar{s}:=\bar{d}-\bar{k}+1$ for simplicity.

\section{The bound}\label{sec:bound}

In this section, we present a lower bound on the number of symbols accessed in any optimal-bandwidth linear repair scheme 
for $(n,k,l)$ MDS linear array codes. Note that any linear repair scheme of linear codes can be realized by a set of dual codewords and the corresponding parity-check equations. Below we show that by instantiating a linear repair scheme using appropriate dual codewords, the MDS property of the code implies a lower bound on the number of accessed symbols.

\begin{theorem}\label{thm:bound}
	Let $\cC$ be an $(n,k,l)$ MDS linear array code over $F$ that supports optimal repair for any \zitan{$h\leq \min\{u,\bar{s}u-v\}$ nodes within a single rack} from any $\bar{d}\geq \bar{k}$ helper racks, each of which contributes the same amount of information for repair. For any linear repair scheme with degree $\bar{d}$, the number of symbols accessed on the helper racks satisfies 
	\begin{align}
		\alpha \geq \frac{\zitan{h} \bar{d}u l}{\bar{s}(u-v)}.\label{eq:rack-access}
	\end{align} If $v>0$ then equality holds if and only if the number of symbols accessed on the node $j=eu+g$ satisfies $\alpha_j=\zitan{h}l/(\bar{s}(u-v))$ for all $e\in\cR$ and $g=0,\ldots,u-1$.
\end{theorem}


\begin{proof}
\zitan{Let $e'$ be the index of the host rack and $j'_1=e'u+g'_1,\ldots,j'_{h}=e'u+g'_{h}$ be the indices of the failed nodes.}
Let $\cR\subset \{0,\ldots,\bar{n}-1\}\setminus\{e'\}$ be the set of $\bar{d}$ helper racks. 
Since $\cC$ supports optimal repair, there exists an $\zitan{h}l\times nl$ matrix $P=[P_0,\ldots,P_{n-1}]$ over $F$ whose rows are formed of dual codewords of $\cC$ such that 
\begin{align}
	\zitan{\rank\, (P_{j'_1},\ldots,P_{j'_{h}}) = hl},\label{eq:rank-failed}\\
	\rank \bar{P}_{e} = \frac{\zitan{h}l}{\bar{s}},\quad e\in\cR,\label{eq:rank-helper}
\end{align} where $P_j,j=0,\ldots,n-1$ are $hl\times l$ matrices over $F$ and $\bar{P}_e:=[P_{eu},\ldots,P_{eu+u-1}]$ is an $\zitan{h}l\times ul$ matrix over $F$. Moreover, $\bar{P}_e=0$ for all $e\in\{0,\ldots,\bar{n}-1\}\setminus(\cR\cup \{e'\})$.
Clearly, we have
\begin{align}
	\bar{P}_{e'} \bar{c}_{e'} = -\sum_{e\in\cR} \bar{P}_e\bar{c}_e.\label{eq:repair}
\end{align} To solve \eqref{eq:repair} for \zitan{$c_{j'_1},\ldots,c_{j'_h}$}, we only need to download $\rank\bar{P}_e=\zitan{h}l/\bar{s}$ symbols from each helper rack $e\in\cR$. At the same time, the number of symbols accessed on rack $e$ to compute $\bar{P}_e\bar{c}_e$ is equal to the number of nonzero columns of $\bar{P}_e$, which we denote by $\acs \bar{P}_e$. 

Let $H=[H_0,\ldots,H_{n-1}]$ be an $rl\times nl$ parity-check matrix of $\cC$ where $H_j,j=0,\ldots,n-1$ are $rl\times l$ matrices over $F$. Note that $H$ generates the dual code $\cC^{\perp}$ of $\cC$, which is an $(n,r,l)$ MDS array code. Consider the set
\begin{align*}
	\cT=\{eu+g\mid e\in \{0,\ldots,\bar{n}-1\}\setminus(\cR\cup\{e'\});g=0,\ldots,u-1\}.
\end{align*} Note that $|\cT|=(\bar{r}-\bar{s})u< r$ where the inequality follows from $1\leq \bar{s}\leq \bar{r}$ and $v<u$. Let $\cC^{\perp}_{\cT}$ be the expurgated code of $\cC^{\perp}$ formed by the codewords of $\cC^{\perp}$ that are zero on $\cT$. Since $\cC^{\perp}$ is an $(n,r,l)$ MDS linear array code over $F$, the expurgated code $\cC^{\perp}_{\cT}$ will form an $(n-|\cT|,r-|\cT|,l)$ MDS array code if one further punctures the coordinates of $\cC^{\perp}_{\cT}$ in $\cT$. Moreover, there exists an $(r-|\cT|)l\times rl$ matrix $A$ over $F$ with $\rank A =(r-|\cT|)l$ such that $\tilde{H}:=AH$ is an $(r-|\cT|)l\times nl$ generator matrix for $\cC^{\perp}_{\cT}$.

Let $\cU$ be a $(u-v)$-subset of $\{0,\ldots,u-1\}$ and let $\cS$ be an $(\bar{s}-1)$-subset of $\cR$. Let $\tilde{e}\in\cR\setminus\cS$ and define 
\begin{align*}
	\cK=\{\tilde{e}u+g\mid g\in\cU\}\cup\{eu+g \mid e\in\cS;g=0,\ldots,u-1\}.
\end{align*}
It is clear that $\zitan{u-v}\leq |\cK|=\bar{s}u-v\leq r$. Moreover, we have $|\cK|=r-|\cT|$.

Observe that by \eqref{eq:rank-failed} we have $\rank P=\zitan{h}l$. 
Since $P$ is a matrix formed by codewords of $\cC^{\perp}$ that are zero on $\cT$ and \zitan{$h\leq \min\{u,\bar{s}u-v\}$}, there exists an $\zitan{h}l \times (r-|\cT|)l$ matrix $B$ over $F$ with $\rank B =\zitan{h}l$ such that $P=B\tilde{H}$.
Let us write $\tilde{H}=[\tilde{H}_0,\ldots,\tilde{H}_{n-1}]$ and let $\tilde{H}_{\cK}$ be the $|\cK|l\times |\cK|l$ matrix over $F$ formed of $\tilde{H}_j,j\in\cK$. Since puncturing the coordinates of $\cC^{\perp}_{\cT}$ in $\cT$ gives rise to an MDS code, we have $\rank \tilde{H}_{\cK} =|\cK|l$.  Let $P_{\cK} = B\tilde{H}_{\cK}$. Then it follows that $\rank P_{\cK} = \rank B\tilde{H}_{\cK} = \rank B = \zitan{h}l$. Meanwhile, we have 
\begin{align}
	\rank P_{\cK}\leq \sum_{e\in\cS}\rank \bar{P}_e + \sum_{g\in\cU}\rank P_{\tilde{e}u+g}.\label{eq:rank-sum}
\end{align} From \eqref{eq:rank-helper}, we have $\rank \bar{P}_e=\zitan{h}l/\bar{s}$ for all $e\in\cS$. Therefore, \eqref{eq:rank-sum} implies $\zitan{h}l/\bar{s}\leq \sum_{g\in\cU}\rank P_{\tilde{e}u+g}$. Noticing $\rank P_j\leq \acs P_j$, we have
\begin{align*}
	\frac{\zitan{h}l}{\bar{s}}\leq \sum_{g\in\cU}\acs P_{\tilde{e}u+g}.
\end{align*}
Summing over $(u-v)$-subsets of $\{0,\ldots,u-1\}$ on both sides of the above inequality, we obtain
\begin{align*}
	\binom{u}{u-v}\frac{\zitan{h}l}{\bar{s}}&\leq \sum_{\substack{\cU\subset\{0,\ldots,u-1\}\\|\cU|=u-v}}\sum_{g\in\cU}\acs P_{\tilde{e}u+g}\\
	&=\binom{u-1}{u-v-1}\sum_{g=0}^{u-1}\acs P_{\tilde{e}u+g}\\
	&=\binom{u-1}{u-v-1} \acs\bar{P}_{\tilde{e}}.
\end{align*}
The above inequality holds for all possible choices of $(\bar{s}-1)$-subset $\cS\subset\cR$ and $\tilde{e}\in\cR\setminus\cS$. Therefore, for any $e\in\cR$ it holds that
\begin{align}
	\acs \bar{P}_{e} \geq \frac{\zitan{h} ul}{\bar{s}(u-v)}.\label{eq:access-bound}
\end{align} Moreover, this bound holds with equality if and only if for every $(u-v)$-subset $\cU\subset\{0,\ldots,u-1\}$ it holds that
\begin{align}
	\sum_{g\in\cU}\acs P_{eu+g}=\frac{\zitan{h}l}{\bar{s}}.\label{eq:access-subset}
\end{align}
Equation \eqref{eq:access-subset} further implies the equality in \eqref{eq:access-bound} holds if and only if $\acs P_{eu+g}=\frac{\zitan{h}l}{\bar{s}(u-v)}$ for every $g=0,\ldots,u-1$. Indeed, this is clear for the case $u-v=1$. Consider the case $u-v>1$ and suppose that \eqref{eq:access-bound} holds with equality while there exist some $g_1,g_2$ such that $\acs P_{eu+g_1}>\frac{\zitan{h}l}{\bar{s}(u-v)}$ and $\acs P_{eu+g_2}<\frac{\zitan{h}l}{\bar{s}(u-v)}$. Since $u-v\leq u-1$, there exist $(u-v)$-subsets $\cU_1,\cU_2$ such that $\cU_1\setminus \cU_2=\{g_1\}$ and $\cU_2\setminus \cU_1=\{g_2\}$. But then $\sum_{g\in\cU_1}\acs P_{eu+g} \neq \sum_{g\in\cU_2}\acs P_{eu+g}$, contradicting \eqref{eq:access-subset}.

In conclusion, the number of symbols accessed on $\bar{d}$ helper racks for repairing \zitan{$h\leq \min\{u,\bar{s}u-v\}$ failed nodes in a single rack} satisfies
\begin{align*}
	\alpha \geq \sum_{e\in\cR}\acs \bar{P}_e \geq \frac{\zitan{h} \bar{d}ul}{\bar{s}(u-v)},
\end{align*} where the equality is attained if and only if for all $e\in\cR$ and $g=0,\ldots,u-1$ it holds that $\alpha_{eu+g}=\acs P_{eu+g}=\frac{\zitan{h}l}{\bar{s}(u-v)}$.
\end{proof}

%
Note that for any $\cS,\cU$ and $\tilde{e}$ defined as in the above proof, it follows from \eqref{eq:rank-sum} that
\begin{align}
	\zitan{h}l\leq \sum_{e\in\cS}\rank \bar{P}_e + \sum_{g\in\cU}\acs P_{\tilde{e}u+g}.\label{eq:trade-off}
\end{align} This inequality suggests a trade-off between the number of symbols downloaded from a helper rack and the number of symbols accessed on a helper node. \zitan{Although in the above proof we optimized the repair bandwidth first and used \eqref{eq:trade-off} to deduce a lower bound on the access complexity, one may first minimize the amount of information accessed for repair and then obtain a lower on the repair bandwidth using \eqref{eq:trade-off}.} For instance, \zitan{consider the case $h=1$}. If we put $\acs P_{\tilde{e}u+g}=\frac{l}{d-k+1}$ where $d=\bar{d}u+u-1$, which is the least access per node according to Proposition~\ref{prop:node-access}, then \eqref{eq:trade-off} implies $\sum_{e\in\cS}\rank \bar{P}_e \geq \frac{(\bar{s}-1)ul}{d-k+1}$. Noticing $\cS$ is an arbitrary $(\bar{s}-1)$-subset of $\cR$, one can deduce that the repair bandwidth in this case is $\beta\geq \frac{\bar{d}ul}{d-k+1}=\zitan{\frac{\bar{d}ul}{\bar{d}u-\bar{k}u+u-v}}$. Thus, in this case, the bandwidth cannot attain the bound \eqref{eq:rack-bound} unless \zitan{$v=0$, i.e., }$u\mid k$. In other words, if an MDS linear array code admits of a linear repair scheme that has the least access for repairing a single node in the rack model, then the scheme does not support optimal repair for the rack model unless $u\mid k$. At the same time, if we apply optimal-access MSR codes for the homogeneous model in the rack model, then they attain the bound \eqref{eq:rack-access-smallest} for access \zitan{in the rack model} and their repair bandwidth is $\frac{\bar{d}ul}{d-k+1}$.\footnote{This was observed in \cite{li2021repair} for the case $\bar{d}=\bar{n}-1$.}

In the following, we define rack-aware optimal-access MSR codes as MDS array codes that attain \zitan{\eqref{eq:rack-bound-h}} and \eqref{eq:rack-access} with equality simultaneously.

\begin{definition}
	Let $\cC$ be an $(n,k,l)$ MDS linear array code over $F$ that supports optimal repair for any \zitan{$h\leq \min\{u,\bar{s}u-v\}$ nodes in a single rack} from any $\bar{d}$ helper racks. Suppose that each of the $\bar{d}$ helper racks provides $\zitan{h}l/\bar{s}$ symbols for the repair of the failed nodes and these symbols are generated by accessing $\zitan{h}l/(\bar{s}(u-v))$ symbols of each node in the rack. Then $\cC$ is said to be a rack-aware MSR code with optimal access.
\end{definition}

As mentioned before, by Proposition~\ref{prop:access-ruv}, the low-access rack-aware MSR codes in \cite{chen2019explicit} turn out to have optimal access \zitan{for repairing a single node} if $\bar{d}=\bar{n}-1$ and $v=u-1$. Theorem~\ref{thm:bound} further implies that the low-access codes actually have the optimal access property \zitan{for repairing a single node} whenever $\bar{d}\leq \bar{n}-1$ and $v=u-1$.

\section{The array code construction}\label{sec:construction}

In this section, we construct a family of rack-aware MSR codes for any $\bar{d}\leq \bar{n}-1$ and any $v\leq u-1$ \zitan{that can repair any single failed node with optimal access.} The size of underlying finite field of the codes is $O(n)$ and the node size is $(\bar{s}(u-v))^{\bar{n}}$. The construction of this new family of codes and its repair scheme can be viewed as an extension to the low-access construction in \cite{chen2019explicit}. In particular, setting $u-v=1$, we obtain the low-access construction. To enable optimal access, we observe that the same set of symbols in helper racks can be utilized to generate $u-v$ distinct linear combinations to be downloaded for repair. This observation, combined with a careful design of parity-check equations, leads to the optimal construction. 

To illustrate the main techniques behind our general code construction, we first present in Section~\ref{sec:example} a simple code $\cC_0$ that supports repairing any single failed node in the \emph{zeroth} rack\footnote{Recall that we label the racks by $0,1,\ldots,\bar{n}-1$.} by downloading $\bar{d}l/\bar{s}$ symbols and accessing $\bar{d}ul/(\bar{s}(u-v))$ symbols, attaining the bounds \eqref{eq:rack-bound-h} and \eqref{eq:rack-access}, respectively. In fact, one can easily modify the construction to obtain a code $\cC_e$ that supports repairing any single failed node in the $e$th rack with the same amount of repair bandwidth and access complexity as $\cC_0$. To construct a general code with optimal access for a single failed node in \emph{any} host rack, we combine the codes $\cC_0,\ldots,\cC_{\bar{n}-1}$ by carefully coupling the symbols in the same rack and expanding the node size. We note that the idea of node size expansion has been used extensively in the literature for constructing high-rate MSR codes.

\zitan{For notational convenience, denote $\eta = u-v$ and $\theta=\bar{s}\eta$.} Let $m=\bar{n}+\bar{s}-1$ and $mu\mid(|F|-1).$ Let $\lambda\in F$ be an element of multiplicative order $mu$. Recall that we write $j={e}u+{g}$ for the index of the node $j=0,\ldots,n-1$, where $0\le {e}< \bar{n}$ and $0\le {g}< u$. For $p=0,\ldots,\zitan{\theta}-1$, let us write $p=\kappa \bar{s}+\tau$ where $0\leq \kappa < \zitan{\eta}$ and $0\leq \tau <\bar{s}$.

\subsection{Repairing any single node in a fixed rack}\label{sec:example}
In this subsection, we present a construction of MDS linear array codes over $F$ that support repair of any single failed node in the zeroth rack from any $\bar{d}$ helper racks. 

\begin{construction} 
	Define an $(n,k=n-r,l=\zitan{\theta})$ array code $\cC_0=\{(c_{j,i})_{0\le j\le n-1; 0\le i\le l-1}\}$ by the following parity-check equations over $F$:
	\begin{align}
		\sum_{j=0}^{n-1}\lambda_j^{t}c_{j,0}+\sum_{g=0}^{u-1}\sum_{p=1}^{\theta-1}\mu_{p,g}^tc_{g,p}=0,\label{eq:defu}\\
		\sum_{j=0}^{n-1}\lambda_j^{t}c_{j,i}=0,\quad i=1,\ldots,\theta-1.\label{eq:defui}
	\end{align}
	where $t=0,\ldots,r-1$, $\lambda_j = \lambda^{{e}+{g}m}$, and 
	\begin{align}
		\mu_{p,g}=
		\begin{cases}
			\lambda^{(g+\kappa)m}, & \tau = 0,\\
			\lambda^{\bar{n}+\tau-1+\kappa m}, & \tau\neq 0.
		\end{cases}\label{eq:mu-simple}
	\end{align}
\end{construction}

Note that for each $i\in\{1,\ldots,\theta-1\}$, the collection of \eqref{eq:defui} for which $t=0,\ldots,r-1$ defines an $(n,k)$ MDS code. Using this fact, it is easy to check that any $k$ nodes of $\cC_0$ suffice to recover the other $n-k$ nodes and thus $\cC_0$ is MDS. 
The repair properties of $\cC_0$ are stated in the following theorem.

\begin{theorem} \label{thm:opt-c0}
	The code $\cC_0$ supports repair of any single failed node in the zeroth rack from any $\bar{d}$ helper racks by downloading from the $\bar{d}$ helper racks $\bar{d}\eta$ symbols of $F$ and these symbols are generated by accessing $\bar{d}u$ symbols in the helper racks. 
\end{theorem}

\begin{proof}
	The repair of any single failed node in the zeroth rack relies only the equations \eqref{eq:defu}.
	To simply the notation, let us assume without loss of generality that $c_0$ is the failed node. 
	
	Let $\cR\subset \{1,\dots,\bar{n}-1\}$ be the set of $\bar{d}$ helper racks and let $\cJ=\{0,\dots,\bar{n}-1\}\setminus \cR$.
	Rearranging \eqref{eq:defu} such that all the information offered by the helper racks appears on one side, we obtain
	\begin{align}
		\sum_{{e}\in{\cJ}}\sum_{{g}=0}^{u-1}\lambda_{{e} u+{g}}^t c_{{e} u+{g},0}
		+
		&\sum_{{g}=0}^{u-1}\sum_{p=1}^{\zitan{\theta}-1}\mu_{p,g}^t c_{{g},p}
		= - 
		\sum_{{e}\in\cR}\sum_{g=0}^{u-1}
		\lambda_{{e}u+{g}}^t c_{{e}u+{g},0}.\nonumber
	\end{align}
	Writing $p=\kappa\bar{s}+\tau$ and further rearranging the above equation, we have
	\begin{align}
			\sum_{g=0}^{u-1}(\lambda_g^{t}c_{g,0}+\sum_{\kappa=1}^{\eta-1}\mu_{\kappa\bar{s},g}^tc_{g,\kappa \bar{s}})
			+\sum_{g=0}^{u-1}\sum_{\kappa=0}^{\eta-1}\sum_{\tau=1}^{\bar{s}-1}\mu_{\kappa\bar{s}+\tau,g}^tc_{g,\kappa \bar{s}+\tau}
			+&\sum_{e\in \cJ\setminus \{0\}}\sum_{g=0}^{u-1} \lambda_{{e} u+{g}}^t c_{{e} u+{g},0}  \nonumber\\
			& =- 
			\sum_{{e}\in\cR}\sum_{g=0}^{u-1}
			\lambda_{{e}u+{g}}^t c_{{e}u+{g},0}.\label{eq:defuk0ex}
	\end{align}
	Denoting the right-hand side of \eqref{eq:defuk0ex} by $\sigma_t$ and using $t=uw+z$, $\lambda_{{e}u+{g}}=\lambda^{{e}+{g}m},\lambda^{um}=1$, and the expression of $\mu_{p,g}$ in \eqref{eq:mu-simple}, we can turn \eqref{eq:defuk0ex} into
	\begin{align}
		\sum_{g=0}^{u-1} \Big(\lambda_{g}^z c_{{g},0}
		+
		\sum_{\kappa=1}^{\zitan{\eta}-1}\mu_{\kappa\bar{s},g}^{z}c_{g,\kappa\bar{s}}\Big)
		+
		&\sum_{{e}\in{\cJ}\setminus \{0\}} \lambda^{{e} uw} \sum_{g=0}^{u-1} \lambda_{eu+g}^z c_{{e} u+{g},0}\nonumber\\
		&+
		\sum_{\tau=1}^{\bar{s}-1}\lambda^{(\bar{n}+\tau-1) uw} 
		\sum_{{g}=0}^{u-1} \sum_{\kappa=0}^{\zitan{\eta}-1}\mu_{\kappa\bar{s}+\tau,g}^{z}c_{g,\kappa\bar{s}+\tau}=\sigma_{uw+z},\label{eq:defuk0exwu}
	\end{align}
	where $w=0,\ldots,\bar{r}-1$ and $z=0,\ldots,{\eta}-1$. Let us write $\cJ=\{e_1=0,e_2,\dots,e_{\bar{n}-\bar{d}}\}$. Define $\alpha_i:=\lambda^{{e}_iu},i=1,\ldots,\bar{n}-\bar{d}$ and $\beta_{\tau}:=\lambda^{(\bar{n}+\tau-1)u}, \tau=1,\ldots,\bar{s}-1$. For each $z$, let us write equations \eqref{eq:defuk0exwu} for all $w=0,\ldots,\bar{r}-1$ in matrix form:
	\begin{equation}\label{eq:ee}
		\underbrace{\left[\begin{array}{ccccccc}
			1                    & \cdots & 1                                    & 1                   & \cdots & 1                             \\
			\alpha_1             & \cdots & \alpha_{\bar{n}-\bar{d}}             & \beta_1             & \cdots & \beta_{\bar{s}-1}             \\
			\vdots               & \ddots & \vdots                               & \vdots              & \ddots & \vdots                        \\
			\alpha_1^{\bar{r}-1} & \cdots & \alpha_{\bar{n}-\bar{d}}^{\bar{r}-1} & \beta_1^{\bar{r}-1} & \cdots & \beta_{\bar{s}-1}^{\bar{r}-1}
		\end{array}
			\right]}_{M}
		\left[\begin{array}{c}
			\sum_{g=0}^{u-1}(\lambda_g^{z}c_{g,0}+\sum_{\kappa=1}^{\eta-1}\mu_{\kappa\bar{s},g}^zc_{g,\kappa \bar{s}}) \\
			\sum_{g=0}^{u-1}\lambda_{e_2u+g}^z c_{e_2u+g,0}\\
			\vdots\\
			\sum_{g=0}^{u-1}\lambda_{e_{\bar{n}-\bar{d}}u+g}^z c_{e_{\bar{n}-\bar{d}}u+g,0}\\
			\sum_{g=0}^{u-1}\sum_{\kappa=0}^{\eta-1}\mu_{\kappa\bar{s}+1,g}^zc_{g,\kappa \bar{s}+1}\\
			\vdots\\
			\sum_{g=0}^{u-1}\sum_{\kappa=0}^{\eta-1}\mu_{\kappa\bar{s}+\bar{s}-1,g}^zc_{g,\kappa \bar{s}+\bar{s}-1}
		\end{array}
		\right]=
		\left[\begin{array}{c}
			\sigma_{z}  \\
			\sigma_{u+z}\\
			\vdots\\
			\sigma_{u(\bar{r}-1)+z}
		\end{array}
		\right].
	\end{equation}
	Clearly, the matrix $M$ in \eqref{eq:ee} is invertible since $\bar{n}+\tau-1>e_i$ for all $\tau$ and $e_i$. So for all $z=0,\ldots,{\eta}-1$ the values in 
	\begin{align}
		&\Big\{\sum_{g=0}^{u-1}\Big(\lambda_{g}^z c_{{g},0}
		+
		\sum_{\kappa=1}^{{\eta}-1}\mu_{\kappa\bar{s},g}^{z}c_{g,\kappa\bar{s}}\Big)
		\Big\},\label{eq:f1}\\
		&\Big\{
		\sum_{{g}=0}^{u-1} \sum_{\kappa=0}^{{\eta}-1}\mu_{\kappa\bar{s}+\tau,g}^{z}c_{g,\kappa\bar{s}+\tau} \mid \tau=1,\ldots,\bar{s}-1 \Big\},\label{eq:f2}\\
		&\Big\{\sum_{{g}=0}^{u-1}\lambda_{eu+g}^z c_{eu+g,0} \mid e\in\cJ\setminus\{0\}\Big\}\nonumber
	\end{align}
	can be found from $\{\sigma_{uw+z}\mid w = 0, \ldots, \bar{r}-1\}$.
	From \eqref{eq:f1}, \eqref{eq:f2} and the local nodes $\{c_g|g=1,\ldots,u-1\}$, we can find
	\begin{align}
		&\Big\{\lambda_{0}^z c_{{0},0}
		+
		\sum_{\kappa=1}^{{\eta}-1}\mu_{\kappa\bar{s},0}^{z}c_{0,\kappa\bar{s}}
		\Big\},\label{eq:f11}\\
		&\Big\{
		\sum_{\kappa=0}^{{\eta}-1}\mu_{\kappa\bar{s}+\tau,0}^{z}c_{0,\kappa\bar{s}+\tau} \mid \tau=1,\ldots,\bar{s}-1 \Big\}.\label{eq:f22}
	\end{align}
	Collecting \eqref{eq:f11} and \eqref{eq:f22} for $z=0,1,\ldots,\eta-1$ and writing them as matrix-vector multiplication, we obtain
	\begin{equation}\label{eq:f111}
		{\left[\begin{array}{cccc}
			1                    & 1                        & \cdots & 1                                \\
			\lambda_{0}          & \mu_{\bar{s},0}          & \cdots & \mu_{(\eta-1)\bar{s},0}          \\
			\vdots               & \vdots                   & \cdots & \vdots                           \\
			\lambda_{0}^{\eta-1} & \mu_{\bar{s},0}^{\eta-1} & \cdots & \mu_{(\eta-1)\bar{s},0}^{\eta-1}
		\end{array}
			\right]}
		\left[\begin{array}{c}
			c_{0,0} \\
			c_{0,\bar{s}}\\
			\vdots\\
			c_{0,(\eta-1)\bar{s}}
		\end{array}
		\right]
		=\left[\begin{array}{c}
			c_{0,0}+\sum_{\kappa=1}^{\eta-1}c_{0,\kappa \bar{s}}\\
			\lambda_0c_{0,0}+\sum_{\kappa=1}^{\eta-1}\mu_{\kappa\bar{s},0}c_{0,\kappa \bar{s}}\\
			\vdots\\
			\lambda_0^{\eta-1}c_{0,0}+\sum_{\kappa=1}^{\eta-1}\mu_{\kappa\bar{s},0}^{\eta-1}c_{0,\kappa \bar{s}}
		\end{array}
		\right],
	\end{equation}
	\begin{equation}\label{eq:f222}
		{\left[\begin{array}{cccc}
			1                     & 1                             & \cdots & 1                                     \\
			\mu_{\tau,0}          & \mu_{\bar{s}+\tau,0}          & \cdots & \mu_{(\eta-1)\bar{s}+\tau,0}          \\
			\vdots                & \vdots                        & \cdots & \vdots                                \\
			\mu_{\tau,0}^{\eta-1} & \mu_{\bar{s}+\tau,0}^{\eta-1} & \cdots & \mu_{(\eta-1)\bar{s}+\tau,0}^{\eta-1}
		\end{array}
			\right]}
		\left[\begin{array}{c}
			c_{0,\tau} \\
			c_{0,\bar{s}+\tau}\\
			\vdots\\
			c_{0,(\eta-1)\bar{s}+\tau}
		\end{array}
		\right]
		=\left[\begin{array}{c}
			\sum_{\kappa=0}^{\eta-1}c_{0,\kappa \bar{s}+\tau} \\
			\sum_{\kappa=0}^{\eta-1}\mu_{\kappa\bar{s}+\tau,0}c_{0,\kappa \bar{s}+\tau}\\
			\vdots\\
			\sum_{\kappa=0}^{\eta-1}\mu_{\kappa\bar{s}+\tau,0}^{\eta-1}c_{0,\kappa \bar{s}+\tau}
		\end{array}
		\right].
	\end{equation}
	Since $\lambda_0$ and $\mu_{\kappa\bar{s}+\tau,0}$, $\kappa=1,\cdots,\eta-1$ are all different and the same holds for $\mu_{\kappa\bar{s}+\tau,0}$, $\kappa=0,1,\cdots,\eta-1$, we can calculate from \eqref{eq:f111} and \eqref{eq:f222} the values of $(c_{0,\kappa\bar{s}+\tau})_{0\leq \kappa < \eta,0\leq \tau<\bar{s}}$. This shows $c_0$ can be recovered from any $\bar{d}$ helper racks. 
	
	During the repair process, the values we need from the helper racks are $\{ \sigma_{uw+z} \mid w=0,\ldots,\bar{r}-1;z=0,\ldots,\eta-1\}$. By definition,
	\begin{align*}
		\sigma_{uw+z}&:=- 
		\sum_{{e}\in\cR}\sum_{g=0}^{u-1}
		\lambda_{{e}u+{g}}^{uw+z} c_{{e}u+{g},0}\\
		& =- 
		\sum_{{e}\in\cR}
		\lambda^{euw}
		\sum_{g=0}^{u-1}
		\lambda_{{e}u+{g}}^{z} c_{{e}u+{g},0},
	\end{align*} where the last equality follows from $\lambda_{eu+g}=\lambda^{e+gm}$ and $\lambda^{mu}=1$. Thus, to repair $c_0$, we need to access the symbols $\{c_{eu+g,0}\mid e\in \cR; g=0,\ldots,u-1\}$ in the nodes of the helper racks and download the symbols $\{\sum_{g=0}^{u-1}
\lambda_{{e}u+{g}}^{z} c_{{e}u+{g},0}\mid e\in \cR;z=0,\ldots,\eta-1\}$. It follows that the number of accessed symbols is $\bar{d}u$ and the number of downloaded symbols is $\bar{d}\eta$, which attain the bound \eqref{eq:rack-access} and \eqref{eq:rack-bound-h}, respectively.
\end{proof}

\subsection{The general construction}\label{sec:general}

We need some more notation to describe the construction.
Let \zitan{$l=\theta^{\bar{n}}$ and let} $i=(i_{\bar{n}-1},\ldots,i_0)$ be the \zitan{$\theta$}-ary representation of $i=0,\ldots,l-1$. \zitan{Namely, we represent the integer $i$ with base $\theta$, and for $0\le a \le \bar{n}-1$, the integer $i_a$ is the $a$th digit of the $\theta$-ary representation of $i$.} Further, let $i(a,b)=(i_{\bar{n}-1},\ldots,i_{a+1},b,i_{a-1},\ldots,i_0)$ \zitan{where $0\le b \le \zitan{\theta}-1$}. \zitan{In other words, $i(a,b)$ is obtained from $i$ by replacing the $a$th digit of the $\theta$-ary representation of $i$ by $b$.} For brevity below we use the notation $$\delta(i):=\mathbbm{1}_{\{i=0\}}.$$

%

\begin{construction} \label{con:oa}
	Define an $(n,k=n-r,l=\zitan{\theta^{\bar{n}}})$ array code $\cC=\{(c_{j,i})_{0\le j\le n-1; 0\le i\le l-1}\}$ by the following parity-check equations over $F$:
	\begin{align}
		\sum_{j=0}^{n-1}\lambda_j^t c_{j,i} + 
		\sum_{j=0}^{n - 1} 
		\delta(i_e) 	
		\sum_{p=1}^{\zitan{\theta}-1}\mu_{p,j}^t c_{j,i({e},p)} 
		=0,\label{eq:oa-code-pc}
	\end{align}
	where $i=0,\ldots,l-1$, $t=0,\ldots,r-1$, $\lambda_j = \lambda^{{e}+{g}m}$, and 
	\begin{align}
		\mu_{p,j}=
		\begin{cases}
			\lambda^{e+(g+\kappa)m}, & \tau = 0,\\
			\lambda^{\bar{n}+\tau-1+\kappa m}, & \tau\neq 0.
		\end{cases}\label{eq:mu_pj}
	\end{align}
\end{construction}

\begin{theorem} \label{thm:opt-single}
	The code $\cC$ given by Construction~\ref{con:oa} is a rack-aware MSR code with optimal access for repairing any single node from any $\bar{d}$ helper racks.
\end{theorem}

\begin{proof}
The proof consists of two parts. We begin with showing that the code $\cC$ allows a linear repair scheme that has both optimal bandwidth and access, and then we complete the proof by showing $\cC$ is MDS.
 
\emph{1) The bandwidth and access properties:}
In the following we present a repair scheme for the code $\cC$ in Construction~\ref{con:oa} for repairing a single node from any $\bar{d}$ helper racks with optimal repair bandwidth and access. Without loss of generality, let us assume for simplicity that $c_{0}$ is the failed node, and thus the index of the host rack is $0$.

As before, $\cR\subset\{1,\ldots,\bar{n}-1\}$ is the set of $\bar{d}$ helper racks. Let $\cJ=\{0,\ldots,\bar{n}-1\}\setminus\cR.$
Note that the index of the host rack is always in $\cJ$, i.e., $0\in\cJ$, and we write $\cJ=\{e_1=0,e_2,e_3,\ldots,e_{\bar{n}-\bar{d}}\}$.
For a given integer $a, 1 \le a \le \bar{n}-\bar{d}$ we will need $a$-subsets of $\cJ$, which we denote by $\cJ_{a}$. We always assume that the index of the host rack is in $\cJ_{a}$. In other words, $\cJ_a$ is an $a$-subset of $\cJ$ such that $0\in\cJ_a$. \zitan{In particular, $\cJ_1=\{0\}$. Let $\cI$ be the set of integers $i\in\{0,\ldots,l\}$ such that $i_0=0$ and} let
$$
\cI_1=\{ i
\in\{0,\ldots,l-1\} \mid i_0 = 0; i_{{e}} \neq 0, {e}\in{\cJ}\setminus{\cJ}_1 \}
$$ 
and define 
$$
\cI_a=\bigcup_{\cJ_a\subseteq{\cJ}}\cI(\cJ_a), \quad a=2,\dots,\bar n-\bar d,
$$
where
$$
\cI(\cJ_a)=\{ i
\in\{0,\dots,l-1\} \mid i_{e}=0,{e}\in\cJ_a; i_{{e}} \neq 0, {e}\in{\cJ}\setminus\cJ_a \}.
$$

We will use the parity-check equations corresponding to $i\in\cI$ and all powers $t=uw+z$ to repair the failed nodes $c_0$ where $w=0,\ldots,\bar{r}-1$ and $z=0,\ldots,\zitan{\eta}-1$, in contrast to the repair schemes in \cite{chen2019explicit} which only involve the parity-check equations corresponding to the powers divisible by $u$. 
Such a set of parity-check equations is well defined since $uw+z\leq r-1$ for any $w\leq \bar{r}-1$ and $z\leq \zitan{\eta}-1$. 

We argue by induction on $a=1,\ldots,\bar{n}-\bar{d}$ to show that the repair is possible. Let us first prove the induction basis by showing that it is possible to recover the values 
\begin{align*}
	&\{c_{0,i(0,p)} \mid p=0,\ldots,\zitan{\theta}-1 \},\\
	&\Big\{\sum_{{g}=0}^{u-1}\lambda_{e_2u+g}^z c_{e_2u+g,i},\ldots,\sum_{{g}=0}^{u-1}\lambda_{e_{\bar{n}-\bar{d}}u+g}^z c_{e_{\bar{n}-\bar{d}}u+g,i} \mid z=0,\ldots,\zitan{\eta}-1\Big\}
\end{align*}
for every $i\in\cI_1$ from the helper racks $\cR$.
 
\zitan{Recall that $\cJ\cup \cR=\{0,\ldots,\bar{n}-1\}$ and that we write $j=eu+g$ for the node index $j$. Rewriting} \eqref{eq:oa-code-pc} for $i\in\cI_1$, we have
\begin{align}
	\sum_{{e}\in{\cJ}}\sum_{{g}=0}^{u-1}\lambda_{{e} u+{g}}^t c_{{e} u+{g},i}
	+
	&\sum_{{g}=0}^{u-1}\sum_{p=1}^{\zitan{\theta}-1}\mu_{p,g}^t c_{{g},i(0,p)}\nonumber\\
	&= - 
	\sum_{{e}\in\cR}\sum_{g=0}^{u-1}\Big(
	\lambda_{{e}u+{g}}^t c_{{e}u+{g},i}
	+
	\delta(i_e)
	\sum_{p=1}^{\zitan{\theta}-1}\mu_{p,eu+g}^t c_{{e}u+{g},i({e},p)}\Big).\label{eq:oa-code-pc-i1}
\end{align}
To shorten our notation, denote the right-hand side of \eqref{eq:oa-code-pc-i1} by $\sigma_{i,w}^{(z)}({\cJ}_1)$, where $uw+z=t$. Note that the value of $\sigma_{i,w}^{(z)}({\cJ}_1)$ only depends on the helper racks. \zitan{Next, recall that $p=\kappa\bar{s}+\tau$. We rearrange the terms on the left-hand side of \eqref{eq:oa-code-pc-i1} depending on whether $\tau=0$.} Then using $t=uw+z$, $\lambda_{{e}u+{g}}=\lambda^{{e}+{g}m}$, $\lambda^{mu}=1$, and \eqref{eq:mu_pj}, \zitan{we can turn \eqref{eq:oa-code-pc-i1} into}
\begin{align}
	\sum_{g=0}^{u-1} \Big(\lambda_{g}^z c_{{g},i}
	+
	\sum_{\kappa=1}^{\zitan{\eta}-1}\mu_{\kappa\bar{s},g}^{z}c_{g,i(0,\kappa\bar{s})}\Big)
	+
	&\sum_{{e}\in{\cJ}\setminus \{0\}} \lambda^{{e} uw} \sum_{g=0}^{u-1} \lambda_{eu+g}^z c_{{e} u+{g},i}\nonumber\\
	&+
	\sum_{\tau=1}^{\bar{s}-1}\lambda^{(\bar{n}+\tau-1) uw} 
	\sum_{{g}=0}^{u-1} \sum_{\kappa=0}^{\zitan{\eta}-1}\mu_{\kappa\bar{s}+\tau,g}^{z}c_{g,i(0,\kappa\bar{s}+\tau)}=\sigma_{i,w}^{(z)}({\cJ}_1),\label{eq:oa-code-pc-i}
\end{align} $i\in\cI_1;w=0,\ldots,\bar{r}-1;z=0,\ldots,\zitan{\eta}-1$.
For $i=1,\ldots,\bar{n}-\bar{d}$ 
define $\alpha_i:=\lambda^{{e}_iu}$ and for $\tau=1,\ldots,\bar{s}-1$ define $\beta_{\tau}:=\lambda^{(\bar{n}+\tau-1)u}$. For each $z$, let us write equations \eqref{eq:oa-code-pc-i} for all $w=0,\ldots,\bar{r}-1$ in matrix form:
\begin{align}\label{eq:oa-ma}
	\begin{bmatrix}
		1 & \cdots & 1 & 1 & \cdots  & 1 \\ 
		\beta_1 & \cdots & \beta_{\bar{s}-1} & \alpha_{1} & \cdots  & \alpha_{\bar{n}-\bar{d}} \\
		\vdots & \ddots & \vdots & \vdots & \ddots & \vdots \\
		\beta_1^{\bar{r}-1} & \cdots & \beta_{\bar{s}-1}^{\bar{r}-1} & \alpha_{1}^{\bar{r}-1} & \cdots  & \alpha_{\bar{n}-\bar{d}}^{\bar{r}-1} 
	\end{bmatrix}
	\begin{bmatrix}
		\sum_{{g}=0}^{u-1} \sum_{\kappa=0}^{\zitan{\eta}-1}\mu_{\kappa\bar{s}+1,g}^{z}c_{g,i(0,\kappa\bar{s}+1)}\\
		\vdots \\
		\sum_{{g}=0}^{u-1} \sum_{\kappa=0}^{\zitan{\eta}-1}\mu_{\kappa\bar{s}+\bar{s}-1,g}^{z}c_{g,i(0,\kappa\bar{s}+\bar{s}-1)}\\
		\sum_{g=0}^{u-1}\Big(\lambda_{g}^z c_{{g},i}
		+
		\sum_{\kappa=1}^{\zitan{\eta}-1}\mu_{\kappa\bar{s},g}^{z}c_{g,i(0,\kappa\bar{s})} \Big)\\
		\sum_{{g}=0}^{u-1}\lambda_{e_2u+g}^z c_{e_2u+g,i} \\
		\vdots \\
		\sum_{{g}=0}^{u-1}\lambda_{e_{\bar{n}-\bar{d}}u+g}^z c_{e_{\bar{n}-\bar{d}}u+g,i}
	\end{bmatrix} =
	\begin{bmatrix}
		\sigma_{i,0}^{(z)}(\cJ_1) \\
		\sigma_{i,1}^{(z)}(\cJ_1) \\
		\vdots  \\
		\sigma_{i,\bar{r}-1}^{(z)}(\cJ_1)
	\end{bmatrix}.
\end{align}
\zitan{Since $\alpha_i,i=1,\ldots,\bar{n}-\bar{d}$ and $\beta_{\tau},\tau=1,\ldots,\bar{s}-1$ are all distinct,} the matrix on the left-hand side of \eqref{eq:oa-ma} is invertible.
Therefore, for all $z=0,\ldots,\zitan{\eta}-1$, the values in 
\begin{align}
	&\Big\{\sum_{g=0}^{u-1}\Big(\lambda_{g}^z c_{{g},i}
	+
	\sum_{\kappa=1}^{\zitan{\eta}-1}\mu_{\kappa\bar{s},g}^{z}c_{g,i(0,\kappa\bar{s})}\Big)
	\Big\},\label{eq:oa-coupled-sum-1}\\
	&\Big\{\sum_{{g}=0}^{u-1}\lambda_{eu+g}^z c_{eu+g,i} \mid e\in\cJ \setminus\{0\}\Big\},\nonumber\\
	&\Big\{
	\sum_{{g}=0}^{u-1} \sum_{\kappa=0}^{\zitan{\eta}-1}\mu_{\kappa\bar{s}+\tau,g}^{z}c_{g,i(0,\kappa\bar{s}+\tau)} \mid \tau=1,\ldots,\bar{s}-1 \Big\}\label{eq:oa-coupled-sum-2}
\end{align}
can be found from the values $\{\sigma_{i,w}^{(z)}(\cJ_1)\mid w = 0, \ldots, \bar{r}-1\}$. Then from \eqref{eq:oa-coupled-sum-1}, \eqref{eq:oa-coupled-sum-2}, and the local nodes $\{c_g\mid g=1,\ldots,u-1\}$ of the host rack, we can further find the values in 
\begin{align}
	&\Big\{
	\lambda_{0}^z c_{0,i}
	+
	\sum_{\kappa=1}^{\zitan{\eta}-1}\mu_{\kappa\bar{s},0}^{z}c_{0,i(0,\kappa\bar{s})} 
	\Big\},\label{eq:oa-coupled-1}\\
	&\Big\{
	\sum_{\kappa=0}^{\zitan{\eta}-1}\mu_{\kappa\bar{s}+\tau,0}^{z}c_{0,i(0,\kappa\bar{s}+\tau)}\mid \tau =1,\ldots,\bar{s}-1\Big\}\label{eq:oa-coupled-2}.
\end{align}
Collecting the values in \eqref{eq:oa-coupled-1} from all $z=0,\ldots,\zitan{\eta}-1$ and writing \zitan{them as matrix-vector multiplication}, we have 
\begin{align}
	\begin{bmatrix}
		1 & 1 & \cdots & 1 \\
		\lambda_{0} & \mu_{\bar{s},0} & \cdots & \mu_{(\zitan{\eta}-1)\bar{s},0}\\
		\vdots  & \vdots & \ddots & \vdots\\
		\lambda_0^{\zitan{\eta}-1} & \mu_{\bar{s},0}^{\zitan{\eta}-1} & \cdots & \mu_{(\zitan{\eta}-1)\bar{s},0}^{\zitan{\eta}-1} 
	\end{bmatrix}
	\begin{bmatrix}
		c_{0,i}\\
		c_{0,i(0,\bar{s})}\\
		\vdots\\
		c_{0,i(0,(\zitan{\eta}-1)\bar{s})}
	\end{bmatrix} 
	&=
	\begin{bmatrix}
		c_{0,i}
		+
		\sum_{\kappa=1}^{\zitan{\eta}-1} c_{0,i(0,\kappa\bar{s})} \\
		\lambda_{0} c_{0,i}
		+
		\sum_{\kappa=1}^{\zitan{\eta}-1}\mu_{\kappa\bar{s},0} c_{0,i(0,\kappa\bar{s})} \\
		\vdots  \\
		\lambda_{0}^{\zitan{\eta}-1} c_{0,i}
		+
		\sum_{\kappa=1}^{\zitan{\eta}-1}\mu_{\kappa\bar{s},0}^{\zitan{\eta}-1}c_{0,i(0,\kappa\bar{s})}
	\end{bmatrix}.\label{eq:oa-ma-local-1}
\end{align}
Similarly, collecting the values in \eqref{eq:oa-coupled-2} from all $z=0,\ldots,\zitan{\eta}-1$ for each fixed $\kappa\bar{s}+\tau$ \zitan{where} $\tau\in\{1,\ldots,\bar{s}-1\}$, we have the following matrix equation
\begin{align}
	\begin{bmatrix}
		1 & 1 & \cdots & 1 \\
		\mu_{\tau,0} & \mu_{\bar{s}+\tau,0} & \cdots & \mu_{(\zitan{\eta}-1)\bar{s}+\tau,0} \\
		\vdots  & \vdots & \ddots & \vdots\\
		\mu_{\tau,0}^{\zitan{\eta}-1} & \mu_{\bar{s}+\tau,0}^{\zitan{\eta}-1} & \cdots & \mu_{(\zitan{\eta}-1)\bar{s}+\tau,0}^{\zitan{\eta}-1}
	\end{bmatrix}
	\begin{bmatrix}
		c_{0,i(0,\tau)}\\
		c_{0,i(0,\bar{s}+\tau)}\\
		\vdots\\
		c_{0,i(0,(\zitan{\eta}-1)\bar{s}+\tau)}
	\end{bmatrix} &=
	\begin{bmatrix}
		\sum_{\kappa=0}^{\zitan{\eta}-1} c_{0,i(0,\kappa\bar{s}+\tau)}\\
		\sum_{\kappa=0}^{\zitan{\eta}-1}\mu_{\kappa\bar{s}+\tau,0} c_{0,i(0,\kappa\bar{s}+\tau)}\\
		\vdots  \\
		\sum_{\kappa=0}^{\zitan{\eta}-1}\mu_{\kappa\bar{s}+\tau,0}^{\zitan{\eta}-1} c_{0,i(0,\kappa\bar{s}+\tau)}
	\end{bmatrix}.\label{eq:oa-ma-local-2}
\end{align}
\zitan{Observe that} the matrices on the left-hand sides of \eqref{eq:oa-ma-local-1} and \eqref{eq:oa-ma-local-2} are invertible. Thus, the values in $\{c_{0,i(0,\kappa\bar{s}+\tau)}\mid \kappa=0,\ldots,\zitan{\eta}-1;\tau=0,\ldots,\bar{s}-1 \}$ can be found from \eqref{eq:oa-coupled-1} and \eqref{eq:oa-coupled-2} for every $i\in\cI_1$. This completes the proof of the induction basis.

Now let us fix $a\in\{2,\dots,\bar n-\bar d\}$ and suppose that for all $i\in\cI_{a'}$ and $1\le a' \le a-1$ we have recovered the values 
\begin{align}
	&\{c_{0,i(0,p)}\mid p=0,\ldots,\zitan{\theta}-1 \},\nonumber\\
	&\Big\{\sum_{{g}=0}^{u-1}\lambda_{e_2u+g}^z c_{e_2u+g,i},\ldots,\sum_{{g}=0}^{u-1}\lambda_{e_{\bar{n}-\bar{d}}u+g}^z c_{e_{\bar{n}-\bar{d}}u+g,i} \mid z=0,\ldots,\zitan{\eta}-1\Big\}\label{eq:oa-induction-hypothesis}
\end{align}
from the information downloaded from the helper racks $\cR.$

Fix a subset $\cJ_a, |\cJ_a|=a,$ and let $i\in \cI(\cJ_a)$. 
From \eqref{eq:oa-code-pc}, we have
\begin{align}
	\sum_{{e}\in{\cJ}}\sum_{{g}=0}^{u-1}\lambda_{{e} u+{g}}^t c_{{e} u+{g},i}
	+
	&\sum_{{e}\in {\cJ_a}}\sum_{{g}=0}^{u-1}\sum_{p=1}^{\zitan{\theta}-1}\mu_{p,eu+g}^t c_{{e} u+{g},i({e},p)}\nonumber\\
	&= -\sum_{{e}\in\cR} \sum_{{g}=0}^{u-1}
	\Big(\lambda_{{e}u+{g}}^t  {c_{{e}u+{g},i}
		+
		\delta(i_e)}
	\sum_{p=1}^{\zitan{\theta}-1}
	\mu_{p,eu+g}^t c_{{e}u+{g},i({e},p)}\Big).\label{eq:oa-code-pc-ia}
\end{align}
Again for notational convenience denote the right-hand side of \eqref{eq:oa-code-pc-ia} by $\sigma_{i,w}^{(z)}( {\cJ_a})$.
Using $t=uw+z$, $p=\kappa\bar{s}+\tau$, $\lambda_{{e}u+{g}}=\lambda^{{e}+{g}m}$, $\lambda^{mu}=1$, and \eqref{eq:mu_pj}, equations \eqref{eq:oa-code-pc-ia} can be turned into
\begin{align}
	\sum_{{e}\in {\cJ_a}}\lambda^{euw} \sum_{{g}=0}^{u-1} \Big( & \lambda_{eu+g}^z c_{{e} u+{g},i}
	+
	\sum_{\kappa=1}^{\zitan{\eta}-1} \mu_{\kappa\bar{s},eu+g}^z c_{{e} u+{g},i({e},\kappa\bar{s})} \Big)
	+ \sum_{{e}\in{\cJ}\setminus\cJ_a} \lambda^{{e} uw} \sum_{{g}=0}^{u-1} \lambda_{eu+g}^z c_{{e} u+{g},i}\nonumber\\
	&+
	\sum_{\tau=1}^{\bar{s}-1} \lambda^{(n+\tau-1)uw}\sum_{{e}\in {\cJ_a}}
	\sum_{{g}=0}^{u-1} \sum_{\kappa=0}^{\zitan{\eta}-1} \mu_{\kappa\bar{s}+\tau,eu+g}^z c_{{e} u+{g},i({e},\kappa\bar{s}+\tau)} 
	=\sigma_{i,w}^{(z)}( {\cJ_a}),\label{eq:oa-code-pc-ia-sim}
\end{align} $i\in\cI(\cJ_a);w=0,\ldots,\bar{r}-1;z=0,\ldots,\zitan{\eta}-1$.
Similarly to what we have above for equations \eqref{eq:oa-code-pc-i}, 
\zitan{for each fixed $z\in\{0,\ldots,\eta-1\}$, the set of $\bar{r}$ equations \eqref{eq:oa-code-pc-ia-sim} with $w=0,\ldots,\bar{r}-1$ can be written in matrix form with an $\bar{r}\times \bar{r}$ invertible matrix}. Therefore, for any $ {\cJ_a}\subseteq{\cJ}$ and every $i\in\cI( {\cJ_a})$, the values in
\begin{align}
	&\Big\{ \sum_{{g}=0}^{u-1} \Big( \lambda_{eu+g}^z c_{{e} u+{g},i}
	+
	\sum_{\kappa=1}^{\zitan{\eta}-1} \mu_{\kappa\bar{s},eu+g}^z c_{{e} u+{g},i({e},\kappa\bar{s})} \Big) \mid e\in\cJ_a \Big\},\label{eq:oa-coupled-sum-3}\\
	&\Big\{ \sum_{{g}=0}^{u-1} \lambda_{eu+g}^z c_{{e} u+{g},i} \mid e\in\cJ\setminus\cJ_a \Big\},\label{eq:oa-step-1}\\
	&\Big\{ 
	\sum_{{g}=0}^{u-1} \sum_{\kappa=0}^{\zitan{\eta}-1} \mu_{\kappa\bar{s}+\tau,g}^z c_{{g},i(0,\kappa\bar{s}+\tau)} +\pi_{i,\tau}^{(z)}(\cJ_a)\mid \tau=1,\ldots,\bar{s}-1\Big\}\label{eq:oa-coupled-sum-4}
\end{align}
can be found from the values $\{\sigma_{i,w}^{(z)}( {\cJ_a})\mid w = 0, \ldots, \bar{r}-1 \}$ for all $z=0,\ldots,\zitan{\eta}-1$, where 
\begin{align*}
	\pi_{i,\tau}^{(z)}(\cJ_a):=\sum_{{e}\in {\cJ_a}\setminus\{0\}}\sum_{{g}=0}^{u-1} \sum_{\kappa=0}^{\zitan{\eta}-1} \mu_{\kappa\bar{s}+\tau,eu+g}^z c_{{e} u+{g},i({e},\kappa\bar{s}+\tau)},\quad \tau=1,\ldots,\bar{s}-1.
\end{align*} 

Next, let us show that $\pi_{i,\tau}^{(z)}(\cJ_a)$ can be found by the previously recovered values in \eqref{eq:oa-induction-hypothesis}. 
Indeed, we have
\begin{align}
	\pi_{i,\tau}^{(z)}(\cJ_a) &=\sum_{{e}\in {\cJ_a}\setminus\{0\}}\sum_{{g}=0}^{u-1}\sum_{\kappa=0}^{\zitan{\eta}-1} \lambda^{(\bar{n}+\tau-1+\kappa m)z} c_{{e} u+{g},i({e},\kappa\bar{s}+\tau)}\nonumber\\
	&= \sum_{\kappa=0}^{\zitan{\eta}-1} \lambda^{(\bar{n}+\tau-1+\kappa m)z} \sum_{{e}\in {\cJ_a}\setminus\{0\}}\sum_{{g}=0}^{u-1}c_{{e} u+{g},i({e},\kappa\bar{s}+\tau)},\label{eq:oa-induction-step}
\end{align} where we used $\mu_{\kappa\bar{s}+\tau,eu+g}=\lambda^{\bar{n}+\tau-1+\kappa m}$ for $\tau\neq 0$. For all $i\in\cI(\cJ_a)$, $e\in\cJ_a\setminus\{0\}$, and $\tau \neq 0$, we have $i(e,\kappa\bar{s}+\tau)\in \cI_{a-1}$. By the induction hypothesis \eqref{eq:oa-induction-hypothesis}, we can calculate $\pi_{i,\tau}^{(z)}(\cJ_a)$ using \eqref{eq:oa-induction-step}. 

Now that $\pi_{i,\tau}^{(z)}(\cJ_a)$ can be computed, we can further determine the following values using \eqref{eq:oa-coupled-sum-3}, \eqref{eq:oa-coupled-sum-4}, and the local nodes $\{c_{\zitan{2}},\ldots, c_{u-1}\}$ of the host rack:
\begin{align}
	&\Big\{ \lambda_{0}^z c_{0,i}
	+
	\sum_{\kappa=1}^{\zitan{\eta}-1} \mu_{\kappa\bar{s},0}^z c_{0,i(0,\kappa\bar{s})} \Big\},\label{eq:oa-coupled-3}\\
	&\Big\{
	\sum_{\kappa=0}^{\zitan{\eta}-1} \mu_{\kappa\bar{s}+\tau,0}^z c_{ 0,i(0,\kappa\bar{s}+\tau)} \mid \tau=1,\ldots,\bar{s}-1\Big\}\label{eq:oa-coupled-4}.
\end{align}
Hence, \eqref{eq:oa-ma-local-1} and \eqref{eq:oa-ma-local-2} also hold for $i\in\cI(\cJ_a),\cJ_a\subset\cJ$. It follows that the values in $\{c_{0,i(0,\kappa\bar{s}+\tau)}\mid \kappa=0,\ldots,\zitan{\eta}-1;\tau=0,\ldots,\bar{s}-1 \}$ can be found from \eqref{eq:oa-coupled-3} and \eqref{eq:oa-coupled-4} for every $i\in\cI_a$. 

\zitan{It remains to show that \eqref{eq:oa-induction-hypothesis} holds for $i\in\cI_a$.
Note that}, for $i\in\cI(\cJ_a)$, $e\in\cJ_a\setminus\{0\}$, and $\kappa\bar{s} \neq 0$, we have $i(e,\kappa\bar{s})\in \cI_{a-1}$. Therefore, $\{\zitan{\sum_{{g}=0}^{u-1}\lambda_{eu+g}^z}c_{eu+g,i(e,\kappa\bar{s})} \mid \kappa=1,\ldots,\zitan{\eta}-1\}$ has been previously recovered for each $e\in\cJ_a\setminus\{0\}$. \zitan{From this set of values, we can calculate 
\begin{align*}
	\sum_{g=0}^{u-1}\sum_{\kappa=1}^{\eta-1} \mu_{\kappa\bar{s},eu+g}^z c_{{e} u+{g},i({e},\kappa\bar{s})} 
	&=\sum_{g=0}^{u-1}\sum_{\kappa=1}^{\eta-1} \lambda_{eu+g}^{z}\lambda^{\kappa mz} c_{{e} u+{g},i({e},\kappa\bar{s})} \\
	&=\sum_{\kappa=1}^{\eta-1} \lambda^{\kappa mz} \sum_{g=0}^{u-1}\lambda_{eu+g}^{z}c_{{e} u+{g},i({e},\kappa\bar{s})} 
\end{align*}
where we used $\mu_{\kappa\bar{s},eu+g}=\lambda^{e+(g+\kappa)m}$ and $\lambda_{eu+g}=\lambda^{e+gm}$.
Therefore, for each $e\in\cJ_{a}\setminus\{0\}$, we can further obtain from \eqref{eq:oa-coupled-sum-3} the values in
\begin{align*}
	\Big\{ \sum_{{g}=0}^{u-1} \lambda_{eu+g}^z c_{{e} u+{g},i} \mid e\in\cJ_a\setminus\{0\} \Big\}.
\end{align*}
Combing the above set of values and \eqref{eq:oa-step-1}, we obtain that \eqref{eq:oa-induction-hypothesis} also holds for $i\in\cI_a$.}
This completes the proof of the induction step. Hence, the values in $\{c_{0,i} \mid i=0,\ldots,l-1\}$ can be recovered from the information obtained from the
helper racks in $\cR.$

Now let us assess the repair bandwidth and access of this repair scheme. It is clear that we need the values of $\sigma_{i,w}^{(z)}( {\cJ_a})$ from the helper racks $\cR$ for the repair of $c_0$. Recall that for all $a=1,\ldots,\bar{n}-\bar{d}$ we have 
\begin{align}
	\sigma_{i,w}^{(z)}( {\cJ_a}) &= - 
	\sum_{{e}\in\cR}\sum_{g=0}^{u-1}\Big(
	\lambda_{{e}u+{g}}^t c_{{e}u+{g},i}
	+
	\delta(i_e)
	\sum_{p=1}^{\zitan{\theta}-1}\mu_{p,eu+g}^t c_{{e}u+{g},i({e},p)}\Big),\label{eq:oa-access}
\end{align}
where $t=uw+z$. 
Therefore, to compute $\{\sigma_{i,w}^{(z)}( {\cJ_a})\mid a=1,\ldots,\bar{n}-\bar{d}\}$ for any $w$ and $z$, we need to access the symbols $\{c_{{e} u+{g},i}\mid 0\le {g}<u,i\in\cI\}$ for each  {${e}\in{\cR}$}. In other words,
we need to access $\zitan{\theta}^{\bar{n}-1}=\frac{l}{\bar{s}(u-v)}$ symbols on each node in the helper racks; thus, the total number of accessed symbols equals
$$\frac{\bar{d}ul}{\bar{s}(u-v)}.$$ This is the smallest possible number according to the bound \eqref{eq:rack-access}, and thus the code supports optimal access. Moreover, the set of symbols we access in each helper rack depends on the index of the host rack but not the index of the helper rack.

To see the repair bandwidth of the scheme, we apply $p=\kappa\bar{s}+\tau$, $\lambda_{{e}u+{g}}=\lambda^{{e}+{g}m}$, $\lambda^{mu}=1$, and \eqref{eq:mu_pj} to \eqref{eq:oa-access} and obtain
\begin{align}
	\sigma_{i,w}^{(z)}( {\cJ_a}) = -\sum_{{e}\in\cR} 
	\Big(
	\lambda^{euw} 
	\Big(
	\sum_{{g}=0}^{u-1}
	&\lambda^{(e+gm)z}  c_{{e}u+{g},i}
	+
	\delta(i_e)
	\sum_{\kappa=1}^{\zitan{\eta}-1}
	\lambda^{\kappa mz}
	\sum_{{g}=0}^{u-1}
	\lambda^{(e+gm)z} c_{{e}u+{g},i({e},\kappa\bar{s})}\Big)\nonumber\\
	&+	
	\delta(i_e)
	\sum_{\tau=1}^{\bar{s}-1}
	\lambda^{\bar{n}+\tau-1}
	\sum_{\kappa=0}^{\zitan{\eta}-1}
	\lambda^{(\bar{n}+\tau-1+\kappa m)z} 
	\sum_{{g}=0}^{u-1} 
	c_{{e}u+{g},i({e},\kappa\bar{s}+\tau)}
	\Big).\nonumber
\end{align}
Therefore, to compute $\sigma_{i,w}^{(z)}( {\cJ_a})$, it suffices to download from each $e\in\cR$ the symbols 
\begin{align*}
	\Big\{ \sum_{{g}=0}^{u-1} \lambda^{(e+gm)z} c_{{e}u+{g},i} \mid z=0,\ldots,\zitan{\eta}-1;i\in\cI\Big\}.
\end{align*}
That is, the symbols downloaded to the host rack $e_1$ from any helper rack $e\in\cR$ are $\zitan{\eta}=u-v$ distinct linear combinations of the values in $\{c_{e u+g,i}\mid g=0,\ldots,u-1\},i\in\cI$. Thus, the total amount of information downloaded for the purposes of repair equals
$$
(u-v)\bar{d}|\cI|=(u-v)\bar{d} \frac{l}{\bar{s}(u-v)}=\frac{\bar{d}l}{\bar{d}-\bar{k}+1}.
$$
This is the smallest possible number according to the bound \eqref{eq:rack-bound}, and thus the code support optimal repair.
\vspace{1ex}

\emph{2) The MDS property:}
It suffices to show that the contents of any $n-r$ nodes suffice to find the values of the remaining $r$ nodes.

Let $\cK=\{j_1,\ldots,j_r\}\subseteq\{0,\ldots,n-1\}$ be the set of $r$ nodes to be recovered from the set of $n-r$ nodes in $\cK^c:=\{0,\ldots,n-1\}\setminus\cK$.
Let us write $j_b={e}_b u+{g}_b$ where $0\le {g}_b< u-1$ for $b=1,\ldots,r$.

Let $ {\cJ}$ be the set of distinct ${e}_b,b=1\ldots,r$. For $1 \le a \le |{\cJ}|$, let $ {\cJ_a}\subseteq{\cJ}$ be such that $| {\cJ_a}| = a$.
Let $ {\cI_0}=\{ i
\in {\{0,\ldots,l-1\}} \mid i_{{e}} \neq 0, {e}\in{\cJ} \}$.
For $1 \le a \le |{\cJ}|$ and $ {\cJ_a}\subseteq{\cJ}$, let $$\cI( {\cJ_a})=\{ i
\in {\{0,\ldots,l-1\}} \mid i_{e}=0,{e}\in {\cJ_a}; i_{{e}'} \neq 0, {e}'\in{\cJ}\setminus {\cJ_a} \}.$$ Let $\cI(a)=\bigcup_{ {\cJ_a}\subseteq{\cJ}}\cI( {\cJ_a})$ where $1\le a \le |{\cJ}|$. Observe that  {the sets $
	\cI_a,0\leq a\leq |\cJ|$ partition the set $\{0,1,\dots,l-1\}$}.

We will prove by induction that we can recover the nodes in $\cK$ from the nodes in $\cK^c$.
First, let us establish the induction basis, i.e., we can recover the values $\{c_{j,i}\mid j\in\cK \}$ for every $i\in {\cI_0}$ from the nodes $\{c_j \mid j\in {\cK^{c}} \}$. 

From \eqref{eq:oa-code-pc}, for $i\in {\cI_0}$, we have
\begin{align}
	\sum_{j\in\cK}\lambda_{j}^t c_{j,i}
	&= -
	\sum_{ {j\in\cK^{c}}}
	 {\Big(
	\lambda_{j}^t c_{j,i}
	+
	\delta(i_e) 	
	\sum_{p=1}^{\zitan{\theta}-1}\mu_{p,j}^t c_{j,i({e},p)}
	\Big)}.
\label{eq:oa-mds}
\end{align}
To simplify notation, denote the right-hand side of \eqref{eq:oa-mds} by $ -\sigma_{i,t}=-\sigma_{i,t}(\emptyset).$
Note that the value of ${\sigma_{i,t}}$ only depends on the nodes $\{c_j\mid  {j\in\cK^{c}}\}$.  {Writing \eqref{eq:oa-mds} in matrix form}, we have
\begin{align*}
	\begin{bmatrix}
	1 & \cdots  & 1 \\
	\lambda_{j_1} & \cdots  & \lambda_{j_{r}} \\
	\vdots & \ddots & \vdots \\
	\lambda_{j_1}^{r-1} & \cdots  & \lambda_{j_{r}}^{r-1}
	\end{bmatrix}
	\begin{bmatrix}
	c_{j_1,i} \\
	c_{j_2,i} \\
	\vdots \\
	c_{j_r,i}
	\end{bmatrix} =
	-\begin{bmatrix}
	 {\sigma_{i,0}} \\
	 {\sigma_{i,1}} \\
	\vdots  \\
	 {\sigma_{i,r-1}}
	\end{bmatrix}.
\end{align*}
Therefore, the values $\{c_{j,i} \mid j\in\cK \}$ can be calculated from the values $\{  {\sigma_{i,t}}\mid t = 0, \ldots, r-1\}$ for every $i\in {\cI_0}$.

Now let us establish the induction step. Suppose we recover the values $\{c_{j,i}\mid j\in\cK \}$ for every $i\in {\cI_{a'}}, 0\le a' \le a-1$ from the nodes $\{c_j\mid  {j\in\cK^{c}} \}$, where $1\le a\le |{\cJ}|$.

{Now let us fix a set $\cJ_{a}\subseteq\cJ$ and let} $i\in \cI( {\cJ_a})$.
From \eqref{eq:oa-code-pc}, we have
\begin{align}
	\sum_{j\in\cK}\lambda_{j}^t c_{j,i}
	& =
	-
	\sum_{p=1}^{\zitan{\theta}-1}
	\sum_{j\in\cK \colon {e}\in {\cJ_a}}
	\mu_{p,j}^t
	c_{j,i({e},p)} -
	\sum_{ {j\in\cK^{c}}}
	\Big(
	\lambda_{j}^t c_{j,i}
	+
	\delta(i_e)
	\sum_{p=1}^{\zitan{\theta}-1}\mu_{p,j}^t c_{j,i({e},p)}
	\Big)\nonumber\\
	& {=:-\rho_{i,t}- \sigma_{i,t}( {\cJ_{a}}),}
	\label{eq:oa-mds-iua}
\end{align}
where the last line serves to introduce the shorthand notation.
Note that we know the values $\{ {\sigma_{i,t}( {\cJ_a})}\mid t=0,\ldots,r-1 \}$ since the value $ {\sigma_{i,t}( {\cJ_a})}$ only depends on the nodes $\{c_j\mid  {j\in\cK^{c}}\}$.
Furthermore, we also know the values $\{ {\rho_{i,t}}\mid t=0,\ldots,r-1 \}$. Indeed, for $i\in\cI( {\cJ_a})$, ${e} \in  {\cJ_a}$, and $p \neq 0$, we have $i({e},p) \in  {\cI_{a-1}}$. By the induction hypothesis, we have recovered the values $\{c_{j,i} \mid i\in {\cI_{a-1}}, j\in\cK \}$, and therefore, we know the values $\{c_{j,i({e},p)} \mid j\in\cK\colon {e}\in {\cJ_a};p\neq 0 \}$ for each $i\in {\cI_{a}}$. It follows that we know the values $\{ {\rho_{i,t}}\mid t=0,\ldots,r-1 \}$.
Writing \eqref{eq:oa-mds-iua} in matrix form, we have
\begin{align*}
	\begin{bmatrix}
	1 & \cdots  & 1 \\
	\lambda_{j_1} & \cdots  & \lambda_{j_{r}} \\
	\vdots & \ddots & \vdots \\
	\lambda_{j_1}^{r-1} & \cdots  & \lambda_{j_{r}}^{r-1}
	\end{bmatrix}
	\begin{bmatrix}
	c_{j_1,i} \\
	c_{j_2,i} \\
	\vdots \\
	c_{j_r,i}
	\end{bmatrix} =
	-\begin{bmatrix}
	 {\rho_{i,0}}+ {\sigma_{i,0}(\cJ_a)} \\
	 {\rho_{i,1}}+ {\sigma_{i,1}(\cJ_a)} \\
	\vdots  \\
	 {\rho_{i,r-1}}+ {\sigma_{i,r-1}(\cJ_a)}
	\end{bmatrix}.
\end{align*}
Therefore, the values $\{c_{j,i} \mid j\in\cK \}$ can be recovered for every $i\in\cI( {\cJ_a})$ and $ {\cJ_a}\subseteq{\cJ}$. It follows that we can recover the values $\{c_{j,i}\mid j\in\cK \}$ for every $i\in {\cI_{a}}$. Thus, all the values $\{c_{j,i} \mid j\in\cK,i\in {\cI_{a}},0\le a\le |{\cJ}|\}=\{c_{j,i}\mid j\in\cK,i\in\{0,\ldots,l-1\}\}$ can be recovered from the nodes $\{c_j\mid  {j\in\cK^{c}}\}$.

Since $\cK$ is arbitrary, we conclude that any $n-r$ nodes can recover the entire codeword, i.e., the code is MDS.
\end{proof}

As a final remark of this section, we note that the idea of utilizing the same set of symbols to generate distinct linear combinations can also be applied to the low-access construction in \cite{chen2019explicit} directly, resulting in a repair scheme that supports optimal repair of up to $u-v$ failed nodes in a single rack by accessing a relatively small number of symbols in helper racks. More precisely, it is possible to show that the rack-aware low-access codes in \cite{chen2019explicit} admit a repair scheme for recovering $h\leq u-v$ failed nodes in a single rack from any $\bar{d}$ helper racks with bandwidth attaining \eqref{eq:rack-bound-h} and the number of symbols accessed for repair is at most $\bar{d}ul/\bar{s}$. An inspection of the bound \eqref{eq:rack-access} reveals that the number of accessed symbols for this scheme is optimal if $h=u-v$. We leave the proof of this statement to interested readers. In fact, we will demonstrate this idea again in Section~\ref{sec:rs} for a different class of codes.

\section{The RS code construction}\label{sec:rs}

In this section, we present a family of scalar MDS codes that support optimal bandwidth and low access for repairing $h\leq u-v$ nodes in a single host rack from any $\bar{d}$ helper racks. Our construction is an extension of the RS codes constructed in \cite{tamo2018repair}, \cite{chen2019explicit}, \cite{chen2020enabling}. As mentioned in the final remark of the previous section, one of the key ingredients of our method is to take advantage of multiple distinct linear combinations on the same set of symbols in helper racks to provide needed information for the repair of failed nodes. The other ingredient essential to the repair scheme of our code construction is a proper basis for representing a symbol in a finite field as a vector over a subfield, which affects the access complexity of the scheme. To construct such a basis, we rely on techniques similar to those in \cite{chen2020enabling}.

We still use the basic notation in the previous sections except that one-based numbering is adopted here to label the node indices and rack indices.
Let $q$ be a power of a prime such that $u\mid (q-1)$ and let $\ff_q$ be the finite field of order $q$. Let $p_1,\dots,p_{\bar n}$ be distinct primes
such that $p_e\equiv 1 \bmod \bar{s}$ for $e=1,\dots,\bar n$; for instance, we can take the {\em smallest} $\bar n$
primes with these properties. For $e=1,\dots,\bar n$ let $\alpha_e$ be a primitive element of $\ff_{q^{p_e}}$. It follows that $\alpha_e$ is an element of degree $p_e$ over ${\mathbb F}_q$. 
Consider the following sequence of algebraic extensions of $\ff_q$: let $K_0=\ff_q$ and for $e=1,\ldots,\bar{n}$ let
\begin{align*}
	{F_e}={K_{e-1}}(\alpha_e),\ K_e=F_e(\beta_e),
\end{align*}
where $\beta_e$ is an element of degree $\bar{s}$ over $F_e$. In the end we obtain the field 
\begin{align*}
	{\mathbb K} := K_{\bar{n}} = \ff_q(\alpha_1,\ldots,\alpha_{\bar{n}},\beta_1,\ldots,\beta_{\bar{n}}).
\end{align*}

The RS code will be constructed over the field $\mathbb{K}$. The evaluation points of the code are chosen from certain cosets of the multiplicative subgroup of order $u$ in $\mathbb{K}^*$. Specifically, let $\lambda\in{\mathbb F}_q$ be an element of multiplicative order $u.$ Define the elements
\begin{align}
	\lambda_{e,g}=\lambda^{g-1}\alpha_{e},\ e=1,\dots,\bar n; g=1,\dots,u.\label{eq:alpha-eg}
\end{align} 
The set of evaluation points $\Omega$ is given by
$$
\Omega=\bigcup_{e=1}^{\bar n} \Omega_e, \text{ where } \Omega_e=\{\lambda_{e,g}\mid g=1,\dots,u\}.
$$
Consider the RS code $\cC=RS_{\mathbb K}(n,k,\Omega):=\{(f(\lambda_{e,g}))_{1\leq e\leq \bar{n};1\leq g\leq u}\mid f\in\mathbb
K[x], \deg f < k\}$. A codeword of $\cC$ has the form $c=(c_1,c_2,\dots,c_n),$ where the coordinate $c_{(e-1)u+g}$ corresponds to the evaluation point $\lambda_{e,g},1\le e\le \bar n, 1\le g\le u$. For notational convenience, we also denote $c=((c_{e,g})_{1\leq e\leq \bar{n};1\leq g\leq u})$.

Before describing the repair scheme for the code $\cC$, let us first establish a few properties of the field tower constructed above. 

\begin{lemma}
	\label{le:ext-degree}
	The extension degrees in the field tower $\ff_q=K_0\subset K_1\subset \ldots \subset K_{\bar{n}}=\mathbb{K}$ are as follows:
	\begin{align*}
		[K_e: \ff_q] = \bar{s}^e\prod_{i=1}^{e}p_i,\ e=1,\ldots,\bar{n}.
	\end{align*} In particular, $l:=[\mathbb{K}: \ff_q] = \bar{s}^{\bar{n}}\prod_{e=1}^{\bar{n}}p_e$.
\end{lemma}
\begin{proof}
	It is clear that the extension degree $[K_e: F_e]=\bar{s}$.
	The field $F_{e}$ is obtained by adjoining $\alpha_e$ to $K_{e-1}$, whose degree over $K_{e-1}$ is $p_e$ by co-primality. Thus, $[K_e: K_{e-1}]=\bar{s} p_e$. It follows that $[K_e: \ff_q] = \bar{s}^e\prod_{i=1}^{e}p_i$ for all $e=1,\ldots,\bar{n}$.
\end{proof}

In the light of Lemma~\ref{le:ext-degree}, we represent $i\in\{0,\ldots,l-1\}$ as the following $2\bar{n}$-tuple:
\begin{align*}
	i=(i_{\bar{n}},i_{\bar{n}-1},\ldots,i_{1},i'_{\bar{n}},i'_{\bar{n}-1},\ldots,i'_{1}),
\end{align*} where $i_e\in\{0,\ldots,p_e-1\}$ and $i'_e\in\{0,\ldots,\bar{s}-1\}$ for $e=1,\ldots,\bar{n}$.

Next, we need a simple proposition.

\begin{proposition}\label{prop:basis}
	For $e=1,\ldots,\bar{n}$, the set $\{1,\alpha_e^{u},\ldots, \alpha_e^{u(p_e-1)}\}$ is a basis for $F_e$ over $K_{e-1}$.
\end{proposition}
\begin{proof}
	Since $u<q$ and $\alpha_e$ is a primitive element in $\ff_{q^{p_e}}$, we have $\deg_{\ff_q} \alpha_{e}^u = \deg_{\ff_q} \alpha_{e} = p_e$. Therefore, the set 
	$\{1,\alpha_e^{u},\ldots, \alpha_e^{u(p_e-1)}\}$ is a basis for $\ff_{q^{p_e}}$ over $\ff_q$. Furthermore, by co-primality,	$\{1,\alpha_e^{u},\ldots, \alpha_e^{u(p_e-1)}\}$ is also a basis for $F_e$ over $K_{e-1}$.
\end{proof}

The construction of the field tower and Proposition~\ref{prop:basis} lead to a natural choice of a basis for $\mathbb{K}$ over $\ff_q$.

\begin{lemma}
	Let 
	\[
	A=\{a_i:=\prod_{x=1}^{\bar{n}}\alpha_x^{ui_x} \prod_{y=1}^{\bar{n}}\beta_{y}^{i'_{y}}  \mid i=0,\ldots,l-1\}.
	\] Then $A$ is basis for $\mathbb{K}$ over $\ff_q$.
\end{lemma}
\begin{proof}
	Since $\deg_{F_{e}}\beta_e =\bar{s}$, the set $\{1,\beta_e,\ldots,\beta_e^{\bar{s}-1}\}$ is a basis for $K_e$ over $F_e$ for $e=1,\ldots,\bar{n}$. On account of Proposition~\ref{prop:basis}, it follows that the elements $a_i,i=0,\ldots,l-1$ form a basis for $\mathbb{K}$ over $\ff_q$.
\end{proof}

The next lemma is the counterpart of Lemma~1 in \cite{tamo2018repair}, properly modified for the rack model. Essentially, it shows for $e=1,\ldots,\bar{n}$, the field $K_e$ can be viewed as a direct sum of $\bar{s}$ vector spaces over $K_{e-1}$ of dimension $p_e$. Furthermore, these $\bar{s}$ vector spaces can be generated by a single vector space via repeatedly multiplying the vector space by $\alpha_e^u$. This structure of $K_e$ will play an important part in optimizing the repair bandwidth of the code $\cC$.

\begin{lemma}
	\label{le:rs-subspaces}
	For $e\in \{1,\ldots,\bar{n}\}$, let 
	\begin{align*}
		E_e =\{
		\beta_e^{j}\alpha_{e}^{u(j+t\bar{s})} \mid j=0,\ldots,\bar{s}-1;t=0,\ldots,(p_e-1)/\bar{s}-1
		\}\cup
		\Big\{
		\alpha_e^{u(p_e-1)}\sum_{j=0}^{\bar{s}-1}\beta_e^{j}
		\Big\}
	\end{align*}
	and define $S_{e}=\Span_{K_{e-1}} E_e$. Then 
	\begin{gather}
		\dim_{K_{e-1}}S_{e}=p_e,\quad  
		S_{e}+S_{e}\alpha_e^{u}+\dots+S_{e}\alpha_e^{u(\bar{s}-1)}=K_e,\label{eq:S}
	\end{gather}
	where $S_e\alpha =\{\xi \alpha \mid \xi\in S_e\}$ and the sum is the Minkowski sum of sets.
\end{lemma}

\begin{proof}
	By Proposition~\ref{prop:basis}, the set 
	$\{1,\alpha_e^{u},\ldots, \alpha_e^{u(p_e-1)}\}$ is a basis for $F_e$ over $K_{e-1}$.
	With this observation, the proof of \cite[Lemma~1]{tamo2018repair} can be followed closely. Namely, one can show $\beta_e^j F_e \subset \tilde{K}_e:=S_{e}+S_{e}\alpha_e^{u}+\dots+S_{e}\alpha_e^{u(\bar{s}-1)}$ for all $j=0,\ldots,\bar{s}-1$ by demonstrating $\beta_e^j\alpha_e^{uj}\{1,\alpha_e^{u},\ldots, \alpha_e^{u(p_e-1)}\}\subset \tilde{K}_e$. The detailed steps are omitted here.
\end{proof} 

In view of Lemma~\ref{le:rs-subspaces}, we can construct another basis for $\mathbb{K}$ over $\ff_q$, given in the next lemma. A proof of the lemma that resembles the proof of \cite[Lemma~10]{chen2020enabling} is presented in Appendix~\ref{app:rep}. 

\begin{lemma}\label{le:rep}
	For $i=0,\ldots,l-1$, let $\cE_i=\{e\in\{1,\ldots,\bar{n}\}\mid (i_e, i'_e)=(p_e-1,\bar{s}-1) \}$ and let 
	\[
	b_i=\prod_{x=1}^{\bar{n}}\alpha_x^{ui_x}\prod_{y\in\cE_i}\Big(\sum_{j=0}^{\bar{s}-1}\beta_{y}^{j}\Big)\prod_{y\notin \cE_i}\beta_{y}^{i'_{y}}.
	\] Then the set $B:=\{b_i\mid i=0,\ldots,l-1\}$ is a basis of $\mathbb{K}$ over $\ff_q$.
	
	Furthermore, for $e=1,\ldots,\bar{n}$, let $A_e=\{a_i\in A\mid (i_e,i'_{e})=(0,0)\}$ and $B_e=\{b_i\in B\mid (i_e,i'_{e})=(0,0)\}$. Then
	\[
	\Span_{\ff_q}A_e = \Span_{\ff_q} B_e.
	\]
\end{lemma}

With these properties of the finite field $\mathbb{K}$ at hand, we are almost ready to describe the repair scheme of the code $\cC$. What is still missing is the representation of the each coordinate $c_{e,g}$, which is a symbol in $\mathbb{K}$, as a vector over $\ff_q$. As one may expect, the representation of each node as an $\ff_q$-vector is crucial to access complexity of the repair scheme.

Our choice of the representation relies on the \emph{trace-dual basis} of the basis $(b_i)$ constructed in Lemma~\ref{le:rep}. Recall that the \emph{trace mapping} $\tr_{\mathbb{K}/\ff_q}$ is given by $x\mapsto 1+x^{|\ff_q|}+x^{|\ff_q|^2}+\cdots+x^{|\ff_q|^{l-1}}$. For a basis $(\gamma_i)$ of $\mathbb{K}$ over $\ff_q$, the trace-dual basis of $(\gamma_i)$ is a basis $(\gamma_i^*)$ for $\mathbb{K}$ over $\ff_q$ that satisfies $\tr_{\mathbb{K}/\ff_q}(\gamma_i\gamma_j^*)=\mathbbm{1}_{\{i=j\}}$ for all $i,j$. As a consequence of this property, for an element $x\in\mathbb{K}$, the coefficients of its expansion in the dual basis $(\gamma_i^*)$ can be found using the basis $(\gamma_i)$ since $x=\sum_{i=0}^{l-1}\tr_{\mathbb{K}/\ff_q}(x\gamma_i)\gamma_i^*$. In other words, the mapping given by $x\mapsto (\tr_{\mathbb{K}/\ff_q}(x\gamma_i))_{0\leq i\leq l-1}$ is bijection. Thus, knowing the coefficients $(\tr_{\mathbb{K}/\ff_q}(x\gamma_i))_{0\leq i\leq l-1}$ suffices to recover the element $x\in\mathbb{K}$. 

The repair scheme of $\cC$ is based on a careful selection of codewords in $\cC^\perp$. Since $\cC^\perp$ itself is a
(generalized) RS code, there is a vector $\omega=(\omega_{e,g})_{1\le e\le\bar n; 1\le g\le u}\in({\mathbb K}^\ast)^n$ such that any codeword of $\cC^\perp$ has the form
$(\omega_{e,g}f(\lambda_{e,g}))_{1\le e\le\bar n; 1\le g\le u},$ where $f\in {\mathbb K}[x]$ is a polynomial of degree at most $r-1.$ Taking the vector $\omega$ into account and using the dual basis $(b_i^{\ast})$ of $(b_i)$, each node $c_{e,g},1\leq e\leq \bar{n},1\leq g\leq u$ is represented by
\begin{align}
	c_{e,g}=\omega_{e,g}^{-1} \sum_{i=0}^{l-1}c_{e,g,i}b_{i}^{\ast}.\label{eq:rs-node}
\end{align} As we will see in Theorem~\ref{thm:rs}, this representation will be instrumental in reducing the access complexity of the repair scheme.


\begin{theorem}\label{thm:rs}
	The code $\cC$ supports optimal repair of $h\leq \eta$ failed nodes in any single rack from any $\bar d$ helper racks. Furthermore, the number of $\ff_q$-symbols accessed on the helper racks is at most $\bar{d}ul/\bar{s}$.
\end{theorem}
\begin{proof}  
	Without loss of generality, assume that
	$c_{1,1},\ldots,c_{1,h}$ are the failed nodes so the index of the host rack is $1$.
	Denote by 
	$\cR\subset\{2,3,\ldots,\bar{n}\},|\cR|=\bar{d}$ the set of helper racks. The repair relies on the information downloaded
	from all the nodes in $\cR$ and the functional nodes in the host rack. Define the annihilator polynomial of the set of
	locators of all the nodes not in the helper racks $\cR$ and the host rack: 
	\begin{align}\label{eq:h}
		h(x) = \prod_{\substack{e\in\{1,\ldots,\bar{n}\}\setminus(\cR\cup\{1\}),\\1\leq g\leq u}} (x-\lambda_{e,g}).
	\end{align}
	Let $t=wu +z$, where $w = 0,\ldots,\bar{s}-1$ and $z=0,\ldots,h-1$. Since 
	\begin{align*}
		\deg x^th(x)\leq (\bar{s}-1)u+\eta-1 +(\bar{n}-\bar{d}-1)u = \bar{r}u-v-1 < r=n-k,
	\end{align*}
	we have \[(\omega_{1,1}\lambda_{1,1}^t h(\lambda_{1,1}),\ldots,\omega_{\bar{n},u}\lambda_{\bar{n},u}^t h(\lambda_{\bar{n},u}) )\in\cC^{\perp}.\]
	So the inner product of this dual codeword and the codeword $c$ is zero. In other words, we have
	\begin{align*}
		\sum_{g=1}^u \omega_{1,g}\lambda_{1,g}^t h(\lambda_{1,g}) c_{1,g} 
		& = -\sum_{e=2}^{\bar{n}}\sum_{g=1}^{u} \omega_{e,g}\lambda_{e,g}^t h(\lambda_{e,g}) c_{e,g}.
	\end{align*}
	Let $E_{e}$ be the set defined in Lemma~\ref{le:rs-subspaces} where $e\in\{1,\ldots,\bar{n}\}$ and let $\xi_{e}\in E_e$.
	For notational convenience, define $D=\{\prod_{e=2}^{\bar{n}}\xi_{e}\prod_{e'=2}^{\bar{n}}\alpha_{e'}^{uw_{e'}} \mid \xi_{e}\in E_e, w_{e'}\in\{0,\ldots,\bar{s}-1\}\}$ and let $\delta\in D$.
	Multiplying both sides of the above equation by $\delta\xi_1$ and evaluating the absolute trace $\tr:=\tr_{\mathbb{K}/\ff_q}$, we obtain
	\begin{align}
		\sum_{g=1}^u\tr\Big(\delta \xi_{1}
		\omega_{1,g}\lambda_{1,g}^t h(\lambda_{1,g}) c_{1,g} \Big)
		& = -\sum_{e=2}^{\bar{n}}\sum_{g=1}^{u}\tr\Big(\delta \xi_{1}
		\omega_{e,g}\lambda_{e,g}^t h(\lambda_{e,g}) c_{e,g} \Big) \label{eq:repair-11}\\
		& = -\sum_{e\in\cR}\sum_{g=1}^{u}\tr\Big(\delta \xi_{1}
		\omega_{e,g}\lambda_{e,g}^t h(\lambda_{e,g}) c_{e,g} \Big)\label{eq:r}\\
		& = -\sum_{e\in\cR}\sum_{g=1}^{u} \lambda^{(g-1)z} 
		\tr\Big(\delta \xi_{1}
		\omega_{e,g}\alpha_{e}^t h(\lambda_{e,g}) c_{e,g} \Big)\label{eq:rack-combo}\\
		& = -\sum_{e\in\cR}\sum_{g=1}^{u} \lambda^{(g-1)z}\sum_{i=0}^{l-1}\tr\Big(\delta \xi_{1}
		\alpha_{e}^t h(\lambda_{e,g}) b_{i}^{\ast} \Big)c_{e,g,i}.\label{eq:node-rep}
	\end{align}
	In the equations above, \eqref{eq:r} follows by \eqref{eq:h}; \eqref{eq:rack-combo} follows by \eqref{eq:alpha-eg} and the fact that $\lambda^{(g-1)t}=\lambda^{(g-1)z}\in\ff_q$; \eqref{eq:node-rep} follows on account of \eqref{eq:rs-node}.
	Note that the left-hand side of \eqref{eq:repair-11} only involves nodes in the host rack while the right-hand of \eqref{eq:node-rep} depends only on the nodes in the helper racks. Therefore, from the content of the functional nodes in the host rack and the helper racks, one is able to find
	\begin{align}
		\Big\{\sum_{g=1}^{h}\tr\Big(\delta \xi_{1}
		\omega_{1,g}\lambda_{1,g}^{uw+z} h(\lambda_{1,g}) c_{1,g} \Big) \mid
		\delta\in D;
		\xi_{1}\in E_{1};
		w=0,\ldots,\bar{s}-1;z=0,\ldots,h-1
		\Big\}.\label{eq:repair-tr}
	\end{align}	
	Next, we show that \eqref{eq:repair-tr} suffices to recover the values of $c_{1,1},\ldots,c_{1,h}$. Indeed, 
	on account of Lemma~\ref{le:rs-subspaces}, the set 
	\begin{align*}
		\Big\{ \delta \xi_{1}
		\alpha_{1}^{uw}  \mid
		\delta\in D;
		\xi_{1}\in E_{1};
		w=0,\ldots,\bar{s}-1
		\Big\}
	\end{align*} is a basis for $\mathbb{K}$ over $\ff_q$. Noticing 
	\begin{align*}
		\sum_{g=1}^{h}\tr\Big(\delta \xi_{1}
		\omega_{1,g}\lambda_{1,g}^{uw+z} h(\lambda_{1,g}) c_{1,g} \Big) &= \tr\Big(\delta \xi_{1} \alpha_{1}^{uw}
		\sum_{g=1}^{h}\lambda^{(g-1)z}\alpha_{1}^{z} \omega_{1,g}h(\lambda_{1,g}) c_{1,g} \Big),
	\end{align*} it follows that the values $\sum_{g=1}^{h}\lambda^{(g-1)z}\alpha_{1}^{z} \omega_{1,g}h(\lambda_{1,g}) c_{1,g}$ can be recovered for each $z=0,\ldots,h-1$, since 
	the mapping
	\begin{equation*}
		\sum_{g=1}^{h}\lambda^{(g-1)z}\alpha_{1}^{z} \omega_{1,g}h(\lambda_{1,g}) c_{1,g}\mapsto \sum_{g=1}^{h}\tr\Big(\delta \xi_{1}
		\omega_{1,g}\lambda_{1,g}^{uw+z} h(\lambda_{1,g}) c_{1,g} \Big),\
		\delta\in D;
		\xi_{1}\in E_{1};
		w=0,\ldots,\bar{s}-1
	\end{equation*}
	is a bijection. Furthermore, observe that 
	\begin{align*}
		\begin{bmatrix}
			1 & 1 & \ldots & 1\\
			\alpha_{1} & \lambda\alpha_{1} & \ldots & \lambda^{(h-1)}\alpha_{1}\\
			\vdots 						 & \vdots						& \ddots & \vdots\\
			\alpha_{1}^{h-1} & \lambda^{h-1}\alpha_{1}^{h-1} & \ldots & \lambda^{(h-1)^2}\alpha_{1}^{h-1}
		\end{bmatrix}
		\begin{bmatrix}
			\omega_{1,1}h(\lambda_{1,1})c_{1,1}\\
			\omega_{1,2}h(\lambda_{1,2})c_{1,2}\\
			\vdots\\
			\omega_{1,h}h(\lambda_{1,h})c_{1,h}
		\end{bmatrix}
		=\begin{bmatrix}
			\sum_{g=1}^{h} \omega_{1,g}h(\lambda_{1,g}) c_{1,g}\\
			\sum_{g=1}^{h}\lambda^{(g-1)}\alpha_{1} \omega_{1,g}h(\lambda_{1,g}) c_{1,g}\\
			\vdots\\
			\sum_{g=1}^{h}\lambda^{(g-1)(h-1)}\alpha_{1}^{h-1} \omega_{1,g}h(\lambda_{1,g}) c_{1,g}
		\end{bmatrix}.
	\end{align*} Since the elements $\lambda^{g-1}\alpha,g=1,\ldots,h$ are distinct, the values 
	\begin{align*}
		\omega_{1,g}h(\lambda_{1,g})c_{1,g}, \ g = 1,\ldots,h
	\end{align*} can be found from the values $\sum_{g=1}^{h}\lambda^{(g-1)z}\alpha_{1}^{z} \omega_{1,g}h(\lambda_{1,g}) c_{1,g},z=0,\ldots,h-1$. 
	Moreover, since $\omega_{1,g}\neq 0$ and $h(\lambda_{1,g})\neq 0$, the values 
	\begin{align*}
		c_{1,g}, \ g = 1,\ldots,h
	\end{align*} can also be found and thus repair of $c_{1,1},\ldots,c_{1,h}$ is possible.
	
	It remains to calculate the amount of bandwidth and access incurred in this repair procedure. For every $\delta\in D,\xi_{1}\in E_{1},w=0,\ldots,\bar{s}-1,$ and $z=0,\ldots,h-1$, the elements 
	\begin{equation}\label{eq:download}
		\sum_{g=1}^{u} \lambda^{(g-1)z}\sum_{i=0}^{l-1}\tr\Big(\delta \xi_{1}
		\alpha_{e}^{uw+z} h(\lambda_{e,g}) b_{i}^{\ast} \Big)c_{e,g,i}
	\end{equation}
	are downloaded from helper rack $e$. Since $\delta\in D$ and $e\neq 1$, by $\eqref{eq:h}$ and the definition of $D$, the element $\delta \alpha_{e}^{uw+z} h(\lambda_{e,g})$ is in $\Span_{\ff_q} A_1$. Thus, by Lemma~\ref{le:rep}, $\delta \alpha_{e}^{uw+z} h(\lambda_{e,g})$ can be written as an $\ff_q$-linear combination of elements in $B_1$. It follows that the element $\delta \xi_{1}	\alpha_{e}^{uw+z} h(\lambda_{e,g}) b_{i}^{\ast}$ can be written as an $\ff_q$-linear combination of the elements in $\xi_1 B_1 := \{\xi_1 b \mid  b \in B_1\}$. By construction of the set $B$, we have $\xi_1 B_1\subset B$. Since $(b_i^*)$ is the dual basis of $B$, at most $|\xi_1B_1|=|B_1|=l/(\bar{s}p_1)$ terms of $\tr(\delta \xi_{1}\alpha_{e}^{uw+z} h(\lambda_{e,g}) b_{i}^{\ast} ),i=0,\ldots,l-1$ evaluate to nonzero. The values of $i=0,\ldots,l-1$ such that $\tr(\delta \xi_{1}\alpha_{e}^{uw+z} h(\lambda_{e,g}) b_{i}^{\ast} )$ does not vanish correspond to the values of $c_{e,g,i}$ accessed for calculating \eqref{eq:download}. Hence, the number of symbols accessed on the helper racks $\cR$ for repair is at most 
	\begin{align}
		u|\cR|\Big|\bigcup_{\xi_1\in E_1}\xi_1 B_1\Big| = u\bar{d}|E_1||B_1|=\frac{\bar{d}ul}{\bar{s}}.
	\end{align}
	Finally, the number of field symbols of $\ff_q$ transmitted in the form of \eqref{eq:download} from the helper racks to the host rack is at most  
	\begin{align}
		h|\cR|\Big|\bigcup_{\xi_1\in E_1}\xi_1 B_1\Big| = \frac{\bar{d}hl}{\bar{s}}.
	\end{align}
	This meets the bound \eqref{eq:rack-bound-h} with equality, and proves the claim of optimal repair.
\end{proof}

\begin{corollary}
	The code $\cC$ supports optimal repair of $\eta$ failed nodes in any single host rack from any $\bar d$ helper racks with optimal access.
\end{corollary}
\begin{proof}
	By Theorem~\ref{thm:rs}, the code $\cC$ supports optimal repair of $\eta$ failed nodes by accessing at most $\bar{d}ul/\bar{s}$ symbols of $\ff_q$ of any $\bar{d}$ helper racks, which is the smallest possible according to the bound \eqref{eq:rack-access} for the case $h=\eta=u-v$.
\end{proof}

\section{Concluding remarks}\label{sec:con}
We have presented a lower bound on the access complexity of linear repair schemes of MDS linear codes that support repairing multiple failed nodes in a single rack with optimal bandwidth in the rack model. By constructing explicit codes, we have also shown this bound is attainable for all admissible parameters when there is a single failed node or when there are exactly $u-v$ failed nodes.  

The problem of repairing multiple failed nodes of MDS codes in the rack model remains largely open. One interesting question for future research is to construct rack-aware MSR codes with optimal access for any $h\leq \min\{u,\bar{s}u-v\}$ failed nodes in a single rack. In fact, it is still unknown how to construct MDS codes in the rack model that can repair any $h\leq \min\{u,\bar{s}u-v\}$ failed nodes in a single rack with optimal repair bandwidth, let alone optimal access. Going beyond a single host rack, another interesting question is to investigate the repair bandwidth when the failed nodes are arbitrarily distributed among multiple host racks.

\appendices
\addtocontents{toc}{\protect\setcounter{tocdepth}{0}}
\section{Proof of Proposition~\ref{prop:rack-bound-h}}\label{app:rack-bound-h}
\begin{proof}
	Let $\cR,|\cR|=\bar{d}$ be the set of helper racks and let $\cI\subset \cR, |\cI|=\bar k-1$ be a subset of the helper racks. 
	Note that the number of effective helper nodes is $d:=\bar{d}u+u-h$ as there are $u$ nodes in each helper rack and $u-h$ functional nodes in the host rack that may participate in the repair process.
	Since $|\cR\setminus\cI|u =(\bar{d}-\bar k+1)u\ge d-k+h$, the MDS property of the code $\cC$ implies that 
	in order for the $d$ effective helper nodes to recover the $h$ failed nodes, 
	the amount of information provided by the racks in $\cR\setminus\cI$ should satisfy
	\begin{equation}\label{eq:l}
		\sum_{i\in \cR\backslash\cI}\beta_{i}\ge hl.
	\end{equation}
	Let us sum the left-hand side on all $\cI\subset\cR, |\cI|=\bar k-1:$
	$$
	\sum_{\begin{substack}{\cI\subset \cR\\|\cI|=\bar k-1}\end{substack}} \sum_{i\in \cR\backslash\cI}\beta_{i}=
	\sum_{i\in\cR}\sum_{\begin{substack}{\cI\subset \cR\\ \cI\not\ni i}\end{substack}}\beta_{i}=\binom{\bar d-1}{\bar k -1}\sum_{i\in \cR}\beta_{i}.
	$$
	Together with \eqref{eq:l} we obtain 
	$$
	\binom{\bar d-1}{\bar k -1}\sum_{i\in \cR}\beta_{i}\ge \binom{\bar d}{\bar k-1}hl
	$$
	or
	$$
	\sum_{i\in \cR}\beta_{i}\ge \frac{h\bar d l}{\bar d-\bar k+1},
	$$
	i.e., \eqref{eq:rack-bound-h}. Moveover, this bound holds with equality if and only if \eqref{eq:l} holds with equality for every $\cI\subset \cR,
	|\cI|=\bar k-1$. We claim that $\eqref{eq:l}$ holds with equality if and only if $\beta_i=hl/\bar s$ for all $i\in\cR$, where $\bar s=\bar d-\bar k+1$. Note that this claim is trivial if $\bar{s}=1$. Suppose $\bar{s}>1$ and for the sake of contradiction that there is a rack $i$ such that 
	$\beta_{i}\ne hl/\bar s,$ for instance, $\beta_i< l/\bar s$. Let $\cJ\subset \cR,|\cJ|=\bar s$, and $i\in \cJ.$ There must be a rack $i_1\in \cJ$ that 
	contributes more than the average number of symbols, i.e., $\beta_{i_1}>hl/\bar s.$ Consider $i_2\in \cR\setminus\cJ$ (which exists since $\bar{k}>1$ implies $|\cJ|<|\cR|$) and the subset $(\cJ\backslash\{i\})\cup \{i_2\}$. By \eqref{eq:l}, we have that $\beta_{i_2}<hl/\bar s.$
	Now for the subset $(\cJ\backslash\{i_1\})\cup\{i_2\}$, \eqref{eq:l} fails to hold with equality, which is a contradiction.
\end{proof}

\section{Proof of Lemma~\ref{le:rep}}\label{app:rep}
\begin{proof}
	Since $|A|=|B|=l=[\mathbb{K}:\ff_q]$, we will show that every element $a_i\in A$ can be written as an $\ff_q$-linear combination of the elements in $B$, thereby implying $B$ is basis for $\mathbb{K}$ over $\ff_q$. 
	This is done by induction on the size of subsets of $\{1,\ldots,\bar{n}\}$. Let $\cJ\subset\{1,\ldots,\bar{n}\}$ and let $A(\cJ)=\{a_i\in A \mid (i_e, i'_e)=(p_e-1, \bar{s}-1),e\in\cJ; (i_e, i'_e)\neq (p_e-1, \bar{s}-1),e\notin\cJ \}$. We argue by induction on $|\cJ|$ that $A(\cJ)$ can be linearly generated by $B$ for any subset $\cJ$. If $|\cJ|=0$, i.e., $\cJ=\emptyset$, then clearly $A(\cJ)\subset B$. Let $0< J\leq\bar{n}$ and assume that, for all $\cJ$ such that $|\cJ|\leq J-1$, the elements in $A(\cJ)$ are linearly generated by the elements in $B$. Now let $i\in\{0,\ldots,l-1\}$ be such that $|\cE_i|=J$. Then we have
	\begin{align*}
		a_i = \prod_{x=1}^{\bar{n}}\alpha_{x}^{ui_x}\prod_{y\in\cE_i}\beta_{y}^{\bar{s}-1}\prod_{y\notin\cE_i}\beta_{y}^{i'_{y}}
	\end{align*}
	and
	\begin{align}
		b_i &= \prod_{x=1}^{\bar{n}}\alpha_x^{ui_x}\prod_{y\in\cE_i}\Big(\sum_{j=0}^{\bar{s}-1}\beta_{y}^{j}\Big)\prod_{y\notin \cE_i}\beta_{y}^{i'_{y}}\nonumber\\
		&= \prod_{x=1}^{\bar{n}}\alpha_x^{ui_x}\Big(
		\sum_{j_{y}=0,y\in\cE_i}^{\bar{s}-1}\prod_{y\in\cE_i}\beta_{y}^{j_{y}}\Big)\prod_{y\notin \cE_i}\beta_{y}^{i'_{y}}.\label{eq:bi-J}
	\end{align}
	Multiplying out the sums on the right-hand side of \eqref{eq:bi-J}, we note that the term with all $j_{y}=\bar{s}-1,y\in\cE_i$ equals $a_i$ while the other terms contain fewer than $|\cE_i|=J$ factors of the form $\alpha_{y}^{u(p_{y}-1)}\beta_{y}^{\bar{s}-1}$.
	Each of such terms belongs to some $A(\cJ)$ with $|\cJ|\leq J-1$, and can be linearly generated by $B$ according to the induction hypothesis. This implies that $a_i$ is also expressible as an $\ff_q$-linear combination of the elements in $B$.
	Therefore, for any $\cJ\subset\{1,\ldots,\bar{n}\},0\leq |\cJ|\leq \bar{n}$, the set $A(\cJ)$ can be linearly generated by $B$. 
	
	To prove the second claim, we observe that $\Span_{\ff_q} B_e\subset \Span_{\ff_q} A_e$ for all $e=1,\ldots,\bar{n}$. Therefore, it suffices to show that every element in $A_e$ can be linearly generated by the elements in $B_e$. The proof follows along the same lines as above and thus is omitted.
\end{proof}

	\bibliographystyle{IEEEtran}
	\bibliography{rack-OA}
\end{document}